\documentclass[11pt]{article}
\usepackage{amsmath,amssymb,amsthm,graphicx}
\DeclareGraphicsExtensions{.pdf}

\usepackage[width=\textwidth,indention=0mm,justification=centerlast,
font=small,labelfont=bf,labelsep=period]{caption}

\usepackage[pdfpagemode=UseOutlines,
bookmarks=true, bookmarksopenlevel=1, bookmarksopen=true,
bookmarksnumbered=true, unicode=true,
pdffitwindow=true,pdfpagelayout=SinglePage,pdfhighlight=/P,
colorlinks=true,linkcolor=blue,citecolor=blue,urlcolor=blue]{hyperref}

\advance\textwidth 20mm  \advance\oddsidemargin -10mm
\advance\textheight 30mm \advance\topmargin -18mm
\allowdisplaybreaks[0]

\theoremstyle{plain}
\newtheorem{theorem}{Theorem}
\newtheorem{proposition}[theorem]{Proposition}
\newtheorem{conjecture}[theorem]{Conjecture}
\theoremstyle{remark}
\newtheorem{remark}{Remark}

\title{Bogoyavlensky lattices and generalized Catalan numbers}

\author{V.E.\:Adler\thanks{L.D.\:Landau Institute for Theoretical Physics, Akademika Semenova av. 1A, 142432, Chernogolovka, Russian Federation (permanent address). E-mail: adler@itp.ac.ru} 
\thanks{Institute of Mathematics, Ufa Federal Research Centre, Russian Academy of Sciences, Chernyshevsky str. 112, 450008, Ufa, Russian Federation}}

\date{February 5, 2022}

\begin{document}\thispagestyle{empty}
\maketitle

\begin{abstract}
We study the problem of the decay of initial data in the form of a unit step for the Bogoyavlensky lattices. In contrast to the Gurevich--Pitaevskii problem of the decay of initial discontinuity for the KdV equation, it turns out to be exactly solvable, since the dynamics is linearizable due to termination on the half-line. The answer is written in terms of generalized hypergeometric functions, which serve as exponential generating functions for generalized Catalan numbers. This can be proved by the fact that the generalized Hankel determinants for these numbers are equal to 1, which is a well-known result in combinatorics. Another method is based on a non-autonomous symmetry reduction consistent with the dynamics. It reduces the lattice equation to a finite-dimensional system and makes it possible to solve the problem for a more general finite-parameter family of initial data.
\medskip

\end{abstract}

%\noindent 2010 Mathematics Subject Classification: 
%37K10, %Completely integrable infinite-dimensional Hamiltonian and Lagrangian systems, integration methods, integrability tests, integrable hierarchies (KdV, KP, Toda, etc.)
%37K35, %Lie-B\"acklund and other transformations for infinite-dimensional Hamiltonian and Lagrangian systems
%37K60, %Lattice dynamics; integrable lattice equations
%34M55, %Painlev\'e and other special ordinary differential equations in the complex domain; classification, hierarchies
%33C20, %Generalized hypergeometric series, ${}_pF_q$
%05A10. %Factorials, binomial coefficients, combinatorial functions

%-------------------------------------------------------------------------------
\section{Introduction}

The Bogoyavlensky lattice BL$_p$, $p\in\mathbb N$, is of the form
\begin{equation}\label{BL}
 u'_j=u_j(u_{j+p}+\cdots+u_{j+1}-u_{j-1}-\cdots-u_{j-p}),\quad j\in{\mathbb Z},
\end{equation}
where $u_j=u_j(t)\in\mathbb R$ and $f'=df/dt$. This integrable equation has been well studied in terms of the inverse scattering method, solitons, conservation laws, Hamiltonian structures, symmetries, and B\"acklund transformations, starting from the pioneering works \cite{Manakov_1974, Kac_Moerbeke_1975} (case $p=1$, the Volterra lattice) and \cite{Narita_1982, Itoh_1987, Bogoyavlensky_1988, Bogoyavlensky_1991}. A detailed exposition and bibliography can be found in the book \cite{Suris_2003}. 

The problem we will study can be described using an intuitive ecological interpretation, the predator-prey model. In it, $u_j$ denotes the population of the species $j$ in suitable units and equation (\ref{BL}) means that the species $j$ feeds on species $j+1,\dots,j+p$ and, accordingly, serves as food for species $j-p,\dots,j-1$. The simplest stationary solution $u_j=1$ describes one of equilibrium states of this model. It is stable with respect to small perturbations, but what happens if at the moment $t=0$ all species with $j\le0$ die out and the food chain breaks? Obviously, this amounts to solving the Cauchy problem with initial data in the form of the step
\begin{equation}\label{step}
 u_j(0)=0,~~ j\le0,\quad u_j(0)=1,~~ j>0.
\end{equation}
It is clear that all variables with $j\le0$ remain at 0. Species 1 has no enemies, but the food remains, so it starts to grow rapidly and devours the species $2,\dots,p+1$. When their populations drop noticeably, the species $p+2$ begin to grow, devouring the species $p+3,\dots,2p+2$, and so on, like falling dominoes. Computational experiment confirms these naive reasoning. It shows that a decay wave is formed in the chain, propagating at a constant speed. After the passage of this wave, the variables with numbers $1+k(p+1)$ grow and tend to the same finite limit (not uniformly in $k$), while the remaining variables tend to zero (Fig.~\ref*{fig:step2} on the left). We only have to find out the quantitative characteristics of these processes. The evolution for $t<0$ is equivalent, due to the symmetry $j\to-j$, $t\to-t$, to the termination of the chain from the right, that is, to the disappearance of food. In this case, all variables tend to zero (Fig.~\ref*{fig:step2} on the right).

\begin{figure}[t]
\centering
\includegraphics[width=0.48\textwidth]{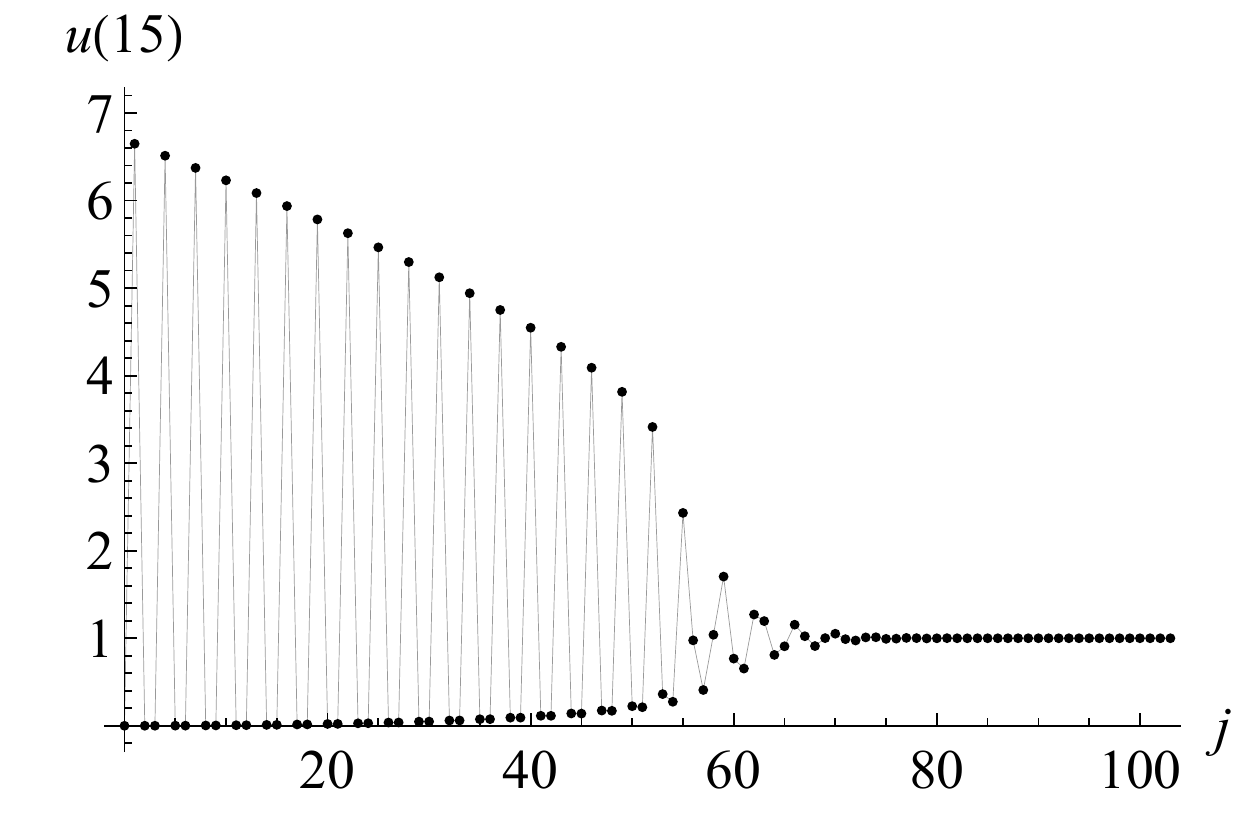}\quad
\includegraphics[width=0.48\textwidth]{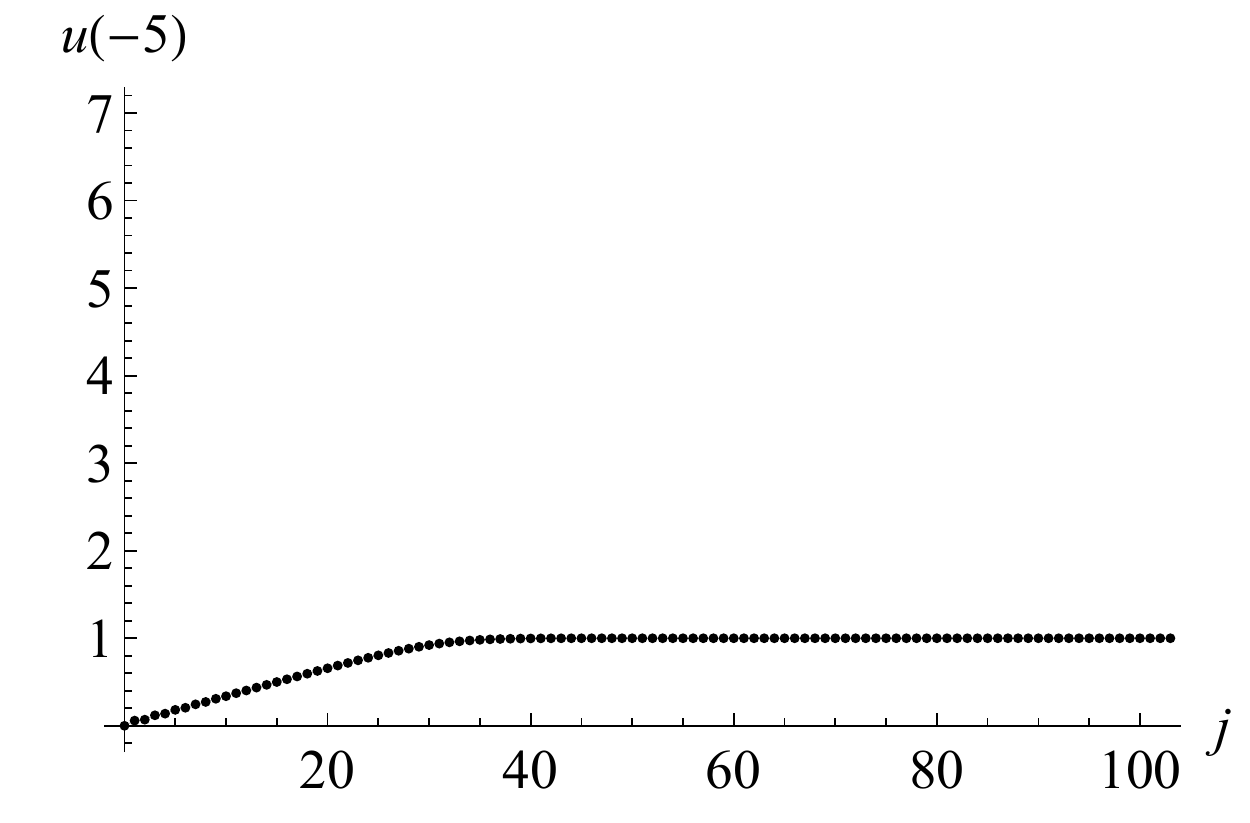}
\caption{Solution of the lattice equation (\ref{BL}) with $p=2$, for the initial data (\ref{step}), at the moments $t=15$ and $t=-5$. The upper limit for $t\to+\infty$ is equal to $(p+1)^{p+1}/{p^p}=27/4$.}
\label{fig:step2} 
\end{figure}

The problem posed is somewhat reminiscent of the Gurevich--Pitaevskii \cite{Gurevich_Pitaevskii_1973} problem on the decay of step-like profile for the KdV equation, but it turns out to be simpler, since the answer is found exactly. We will prove the following result.

\begin{proposition}\label{pr:step}
The solution of the Cauchy problem (\ref{BL}), (\ref{step}) is given by the formulae $u_1=f'_1/f_1$, $u_j=f'_j/f_j-f'_{j-1}/f_{j-1}$, where $f_j$ for $j=1,\dots,p+1$ are equal to the generalized hypergeometric functions
\begin{equation}\label{stepf}
 f_j= {}_pF_p\left(\frac{j}{p+1},\dots,\widehat1,\dots,\frac{j+p}{p+1};\; 
     \frac{j+1}{p},\dots,\frac{j+p}{p};\;\frac{(p+1)^{p+1}}{p^p}t\right),
\end{equation}
where $\widehat1$ denotes the excluded value in the first group of parameters (that is, the numerators of the fractions change from $j$ to $j+p$ with $p+1$ omitted).
\end{proposition}

Let us explain that in order to determine the solution on the half-line, it suffices to specify $u_1,\dots,u_p$, since the variables with $j>p$ are calculated recursively directly from the lattice equations. The function $f_{p+1}$ has the same form as the previous ones, but further formulae become more complicated; they are given in Sect.~\ref{s:det} and are written in terms of Wronskians formed according to certain rules from the functions $f_1,\dots,f_p$ and their derivatives.
 
In the simplest case $p=1$, the problem was studied in papers \cite{Kulaev_Shabat_2019, Adler_Shabat_2018, Adler_Shabat_2019}. The approach from \cite{Adler_Shabat_2018} is based on the remarkable observation that $f_1$ serves as the exponential generating function (EGF) for the Catalan numbers $c_n$. For $p=1$, all $u_j$ are expressed in terms of the Hankel determinants of two types, depending on the parity of $j$:
\begin{equation}\label{f1}
 f_{2k+1}=\left|\begin{matrix}
  f_1 & \dots & f^{(k)}_1  \\
  \hdotsfor[2]{3} \\
  f^{(k)}_1 & \dots & f^{(2k)}_1
 \end{matrix}\right|,\quad
 f_{2k+2}=\left|\begin{matrix}
  f'_1 & \dots & f^{(k+1)}_1  \\
  \hdotsfor[2]{3} \\
  f^{(k+1)}_1 & \dots & f^{(2k+1)}_1
 \end{matrix}\right|,\quad k\ge0.
\end{equation}
For $t=0$, the derivatives are replaced with the Catalan numbers $c_n=f^{(n)}_1(0)$ and Proposition \ref{pr:step} follows from the well-known result that all such determinants are equal to 1 \cite{Aigner_1999, Stanley_1999, Layman_2001}. Such identities (that is, the calculation of the Hankel determinants for various sequences) have been actively studied in combinatorics for a long time; we mention only a few works on this topic \cite{Radoux_1979, Radoux_2002, Ehrenborg_2000, Mays_Wojciechowski_2000, Krattenthaler_2005, Krattenthaler_2010, Spivey_Steil_2006, Cvetkovic_Rajkovic_Ivkovic_2002, Rajkovic_Petkovic_Barry_2007, Peart_Woan_2000,Petkovic_Barry_Rajkovic_2012}. They are related to the problems of counting paths, graphs, trees, triangulations, etc.~and use various analytical methods, including integral transformations, continued fraction expansions of generating functions, and matrix factorization.

For a general $p$, functions $f_1,\dots,f_p$ serve as the EGFs for the generalized Catalan numbers introduced in \cite{Hilton_Pedersen_1991}, see also \cite{Graham_Knuth_Patashnik_1990} (the names of the Fuss--Catalan or the Raney sequences are also used in the literature). The variables $u_j$ are expressed by Wronskians formed of these functions is such a way that each row is obtained from the previous one by shifting to the left by $p$ positions rather than 1 as in (\ref{f1}). Determinants of such type were studied in papers \cite{Chamberland_French_2007, Krattenthaler_2010}. In particular, a rather general formula obtained in \cite{Krattenthaler_2010} implies an identity proving the Proposition \ref{pr:step}.

The emergence of combinatorial objects in problems of mathematical physics has been repeatedly noted. Let us mention only a few close examples: a connection of the Catalan numbers with the Chebyshev polynomials \cite{Artisevich_Bychkov_Shabat_2020, Bychkov_Shabat_2021}, the Catalan and Hurwitz numbers in the theory of the dispersionless Toda chain \cite{Kodama_Pierce_2009, Takasaki_2018}, the generalized Catalan numbers and spectral properties of products of random matrices \cite{Penson_Zyczkowski_2011}. The solution of our problem described above is another example of such a connection with combinatorics. This solution is very simple, but at the same time it gives some dissatisfaction, since it is based on some determinant identity, the nature and proof of which are not very clear to non-experts. In fact, the complexity of the problem is only moved to another place (in \cite{Krattenthaler_2010}, the technique of counting paths on a lattice is used, and the proof is based on some previous results and is not self-sufficient). Instead, it would be nice to have a solution which explains, in a language closer to the original equations, how the functions (\ref{stepf}) appear in the answer.

In the case $p=1$, such an alternative solution was proposed in \cite{Adler_Shabat_2019}. It uses a non-autonomous reduction consistent with the dynamics defined by the lattice equation and related to its master-symmetry. Thus, our problem finds its place among other problems related with the additional symmetry algebra and the string equations, which is a long-studied direction in the theory of integrable equations; we only mention papers \cite{Fokas_Its_Kitaev_1991, Chen_Hu_Muller-Hoissen_2018} which are also related to the Volterra lattice. Our reduction turns this lattice equation into a finite-dimensional system and leads to a family of solutions expressed in terms of the Painlev\'e transcendents, and the additional termination on the half-line leads to special solutions in terms of ${}_1F_1$. The set of initial data for this family is described by an explicit formula and it contains the unit step as a particular example, which gives an alternative proof of Proposition \ref{pr:step} for $p=1$. In this solution, the Wronskian formulae are not needed at all. However, if we take them into account then this proves the above identity about the Hankel determinants for the Catalan numbers, that is, the connection with combinatorics works in the opposite direction as well.

The purpose of this article is to generalize this method for arbitrary $p$. There are difficulties with the use of master-symmetry, because for $p>1$ it is non-local \cite{Zhang_Tu_Oevel_Fuchssteiner_1991, Wang_2012} and leads to overly complicated reductions. Nevertheless, it turns out that there is a rather simple reduction with the required properties
for the modified Bogoyavlensky lattice. In general, this reduction is related to some higher analogues of the Painlev\'e equations and the termination on the half-line leads to the generalized hypergeometry.

The outline of the paper is as follows. In Sect.~\ref{s:Taylor} we present experiments with Taylor series for solving the Cauchy problem (\ref{BL}), (\ref{step}). Combined with the use of such an extensive database as the On-Line Encyclopedia of Integer Sequences (OEIS) \cite{OEIS}, this makes possible to reveal the relation with the generalized Catalan numbers and, as a corollary, to guess the answer from Proposition \ref{pr:step}. We provide the necessary information about the generalized Catalan numbers, including the derivation of exponential generating functions, as well as an example related to the Toda lattice and the Narayana polynomials \cite{Bonin_Shapiro_Simion_1993,Sulanke_2000,Lassalle_2012,Petkovic_Barry_Rajkovic_2012}. 

Sect.~\ref{s:det} contains a proof of Proposition \ref{pr:step} based on comparison of the Wronskian formulae for the lattice equation (\ref{BL}) with identities for the generalized Catalan numbers from \cite{Krattenthaler_2010}.

In Sect.~\ref{s:p1}, we describe the master-symmetry reduction from \cite{Adler_Shabat_2019} which provide an alternative proof for the case of the Volterra lattice ($p=1$). Sect.~\ref{s:p1q} contains a similar reduction with one extra parameter for the Toda lattice. A well-known substitution between the Toda and Volterra lattices makes possible to obtain for the latter the solution of the Cauchy problem with the initial data in the form of two-level step (with different values for odd and even $j$). Under this deformation, the Catalan numbers are replaced with the Narayana polynomials. 

The main results of the work are contained in Sect.~\ref{s:p} devoted to the case of arbitrary $p$. The reduction method allows us to prove Theorem \ref{th:p} which describes the solution of the Cauchy problem for (\ref{BL}) with certain $2p$-parametric family of initial data. Proposition \ref{pr:step} is a particular case of this theorem. The use of exact formulae (\ref{stepf}) makes it possible to calculate the asymptotics for the variables $u_j$ with respect to $t$ and to obtain some numbers characterizing the decay of the unit step.

%-------------------------------------------------------------------------------
\section{Solutions in the form of Taylor series}\label{s:Taylor}

\subsection{Catalan numbers}

Heuristic arguments leading to Proposition \ref{pr:step} can be obtained by the following computational experiment. Let us construct $u_j$ in the form of Taylor series at $t=0$. The derivative $u^{(n)}_j$ is calculated by virtue of equation (\ref{BL}) as a polynomial in $u_{j-pn},\dots,u_{j+pn}$, after that it remains to substitute the initial data at $t=0$. This is not very efficient, but one can easily calculate about 10 coefficients. Naturally, for arbitrary numerical values of $u_j(0)$ nothing reasonable can be expected. The unit step (\ref{step}) turns out to be `magic' initial data which yield something known. For $p=1$, the corresponding values $u^{(n)}_1(0)$, $n\ge0$ are
\[
 1, 1, 1, 0, -4, -10, 15, 210, 504, -3528, -34440,\dotso.
\]
This is the sequence A302197 from OEIS, the Hurwitz logarithm of Catalan numbers, which means that
\begin{equation}\label{u1f}
 u_1= (\log f_1)',
\end{equation}
where $f_1$ is the EGF for the Catalan numbers $c_n$ ($1,1,2,5,14,42,\dots$). This function is known:
\begin{equation}\label{egf-c}
 f_1= c_0+c_1t+c_2\frac{t^2}{2!}+\dots+c_n\frac{t^n}{n!}+\dots={}_1F_1(\tfrac{1}{2},2,4t).
\end{equation}
Thus, if the derivatives of $u_1$ really match A302197 then the formulae (\ref{u1f}) and (\ref{egf-c}) give the solution of our Cauchy problem. All we need to get the answer is to calculate several coefficients of the Taylor series in the first node of the lattice and compare with the available database!

\subsection{Generalized Catalan numbers}

Let us repeat the calculations for $p=2$. The derivatives $u^{(n)}_1(0)$ form the sequence
\[
 1, 2, 5, 10, -8, -255, -1587, -1862, 76944,\dots,
\]
which is not in OEIS. However, if, as before, we pass to the series $f_1$ according to the formula (\ref{u1f}) then it turns out to be the EGF for the A001764 sequence, which belongs to the Fuss--Catalan family. This is also true in the general case. It turns out that well-known combinatorial sequences arise not for $u_j$ themselves, but for $\tau$-functions of the lattice, which we denote by $f_j$. For the lattice on the half-line (with general initial conditions, not just in the form of the unit step), these variables are introduced as follows.

\begin{proposition}\label{pr:uf}
The general solution of the lattice equation (\ref{BL}) terminated by boundary conditions $u_j=0$ for $j\le0$ and such that $u_j\ne0$ for $j>0$ admits the representation
\begin{equation}\label{uf}
 u_1=f'_1/f_1,\quad u_j=f'_j/f_j-f'_{j-1}/f_{j-1},~~ j>1,
\end{equation}
where $f_j$ satisfy the bilinear lattice equation
\begin{equation}\label{fL}
 f_{j-1}f'_j-f'_{j-1}f_j=f_{j-p-1}f_{j+p},\quad j=1,2,\dots,
\end{equation}
with functions $f_{-p},\dots,f_p$ such that
\begin{equation}\label{fic}
 f_{-p}=\dots=f_0=1,\quad f_1(0)=\dots=f_p(0)=1.
\end{equation}
\end{proposition}

\begin{proof} 
Assuming $u_j\ne0$, we put (\ref{BL}) in the form
\[
 (\log u_j)'=P(u_j),
\]
where $P$ is the difference operator ($T$ denotes the shift operator $j\mapsto j+1$)
\[
 P=T^p+\cdots+T-T^{-1}-\cdots-T^{-p}=(T^{p+1}-1)(1-T^{-p})(T-1)^{-1}.
\]
This equation can be resolved by setting
\[
 \log u_j = (T^{p+1}-1)(1-T^{-p})(\log f_{j-1}),\quad (T-1)(\log f_{j-1})'=u_j.
\]
This is equivalent to
\begin{equation}\label{uff}
 u_j=\frac{f'_j}{f_j}-\frac{f'_{j-1}}{f_{j-1}}=\frac{f_{j-p-1}f_{j+p}}{f_{j-1}f_j},
\end{equation}
which gives both equation (\ref{fL}) and the substitution to (\ref{BL}). To restrict the lattice on the half-line, it is necessary to extend these relations for the zero values $u_{-p+1}=\dots=u_0=0$, whence it follows that $f_{-p},\dots,f_0$
coincide up to numerical factors. The equations for $f_j$ do not change under the gauge transformation
\[
 \widetilde f_j=\gamma(t)\gamma_jf_j,\quad \gamma_{j-p-1}\gamma_{j+p}=\gamma_{j-1}\gamma_j,\quad 
 \gamma(t)\ne0,\quad \gamma_j\ne0,
\]
and we can set $f_{-p}=\dots=f_0=1$ without loss of generality, by suitable choice of $\gamma(t)$ and $\gamma_{-p},\dots,\gamma_0$. After that, we can adopt the normalization $f_1(0)=\dots=f_p(0)=1$ by
choosing $\gamma_1,\dots,\gamma_p$.
\end{proof}

A solution of equation (\ref{fL}) (and (\ref{BL}) with it) is uniquely constructed from given $f_1,\dots,f_p$, assuming that all $f_j$ for $j>p$ do not vanish identically (which can happen for a special choice of initial functions). Expressions for subsequent $f_j$ involve the determinant formulae from Sect.~\ref{s:det}, but at the moment they are not needed, since we only want to calculate the Taylor series for $f_1,\dots,f_p$. Technically, it is more convenient to do this using one more auxiliary equation instead of the bilinear lattice, the modified Bogoyavlensky lattice \cite{Bogoyavlensky_1988}
\begin{equation}\label{mBL}
 v'_j=v^2_j(v_{j+p}\cdots v_{j+1}-v_{j-1}\cdots v_{j-p})
\end{equation}
related with (\ref{BL}) and (\ref{fL}) by substitutions
\begin{equation}\label{uvf}
 u_j=v_j\cdots v_{j+p},\quad v_j=\frac{f_{j-p-1}f_j}{f_{j-p}f_{j-1}}.
\end{equation}
Under the boundary condition (\ref{fic}) we have $v_0=0$ and $v_j=f_j/f_{j-1}$ for $j=1,\dots,p$, that is
\[
 f_1=v_1,\quad f_2=v_1v_2,~~\dots,~~ f_p=v_1\cdots v_p.
\]
The derivatives $f^{(n)}_j(0)$ of these products are straightforwardly calculated in virtue of equation (\ref{mBL}) from given initial data for the variables $v_j$, and it is easy to see that the unit step (\ref{step}) corresponds to initial data of the same form $v_j$: $v_0(0)=0$, $v_j(0)=1$, $j>0$. As the result, we obtain the Table \ref{t:Cpjn} and comparison with the OEIS suggests that these are nothing but the generalized Catalan numbers 
\begin{equation}\label{Cpjn}
 C^p_{j,n}=\frac{j}{pn+j}\binom{(p+1)n+j-1}{n}=\frac{j}{(p+1)n+j}\binom{(p+1)n+j}{n},
\end{equation}
where $\binom{m}{n}=\frac{m!}{n!(m-n)!}$. For sequences in the rows of the table, the OEIS provides expressions for the general term, generating functions and numerous combinatorial and algebraic properties. Thus, we come to the following experimental fact.

\begin{table}[t]
\begin{gather*}
 \begin{array}{c|c|rrrrrrrrrr}
  p & j \backslash n & 0 & 1 & 2 & 3 &  4 &  5 & 6 & 7 & & \text{OEIS}~~\\
  \hline
  1 & 1 & 1 & 1 &  2 &  5 &  14 &   42 &   132 &    429 & \dots & \text{A000108}\\
  \hline
  2 & 1 & 1 & 1 &  3 & 12 &  55 &  273 &  1428 &   7752 & \dots & \text{A001764}\\
    & 2 & 1 & 2 &  7 & 30 & 143 &  728 &  3876 &  21318 & \dots & \text{A006013}\\
  \hline
  3 & 1 & 1 & 1 &  4 & 22 & 140 &  969 &  7084 &  53820 & \dots & \text{A002293}\\
    & 2 & 1 & 2 &  9 & 52 & 340 & 2394 & 17710 & 135720 & \dots & \text{A069271}\\
    & 3 & 1 & 3 & 15 & 91 & 612 & 4389 & 32890 & 254475 & \dots & \text{A006632}\\
  \hline
  4 & 1 & 1 & 1 &  5 & 35 & 285 & 2530 & 23751 & 231880 & \dots & \text{A002294}
 \end{array}\\[-8pt] \dots\dots\dots\dots 
\end{gather*} 
\caption{Generalized Catalan numbers $C^p_{j,n}$.}\label{t:Cpjn}
\end{table}

\begin{conjecture}\label{conj:fCpjn}
Functions $f_j$ corresponding to the initial data in the form of the unit step are such that
\begin{equation}\label{fCpjn}
 f^{(n)}_j(0)=C^p_{j,n} \quad\text{for}~~ p=1,2,\dots,~~ j=1,\dots,p+1,~~ n=0,1,\dots,
\end{equation}
that is, $f_1,\dots,f_{p+1}$ serve as the EGFs for the generalized Catalan numbers.
\end{conjecture}

The notation is slightly inconsistent: strictly speaking, we should attach the index $p$ to variables $u_j$, $v_j$ and $f_j$ as well, but this clutters up equations. We assume that $p$ is the global parameter and all objects may depend on it even when it is not indicated explicitly.  

For reference, we note that the first row of the table ($p=j=1$) contains the classical Catalan numbers, the rows $j=1$ for $p>1$ are known as the Fuss--Catalan or the Raney sequences \cite{Graham_Knuth_Patashnik_1990}, and the general definition appeared for the first time, apparently in the work \cite{Hilton_Pedersen_1991}. There are no unified notation for $C^p_{j,n}$ in the literature. Moreover, in many articles, the parameter $p'=p+1$ is used instead of $p$. One of the many combinatorial interpretations of these numbers is the following: $C^p_{j,n}$ is the number of paths which lead from $(0,0)$ to $(pn+j,n)$ and lie strictly below the line $j=pn$, if it is permitted to move only one unit to the right or up (see Fig.~\ref{fig:path}).

The most important property for us is that the EGFs of the sequences $C^p_{j,n}$ are expressed in terms of generalized hypergeometric functions
\begin{equation}\label{pFq}
 {}_pF_q(a_1,\dots,a_p;b_1,\dots,b_q;z)=\sum^\infty_{n=0}\frac{(a_1)_n\cdots(a_p)_n}{(b_1)_n\cdots(b_q)_n}\frac{z^n}{n!},
\end{equation}
where $(a)_0=1$ and $(a)_n=a(a+1)\cdots(a+n-1)$ (the rising factorial). Recall, that for $p<q+1$ these series are convergent for all $z\in\mathbb C$, that is, they define entire functions of $z$ (see e.g.~\cite{NIST}). In our case $p=q$.

The following statement implies that Conjecture \ref{conj:fCpjn} and Proposition \ref{pr:step} are equivalent. These formulae are well-known, but we provide the proof for the sake of completeness. 

\begin{figure}[t]
\centering
\includegraphics[width=0.6\textwidth]{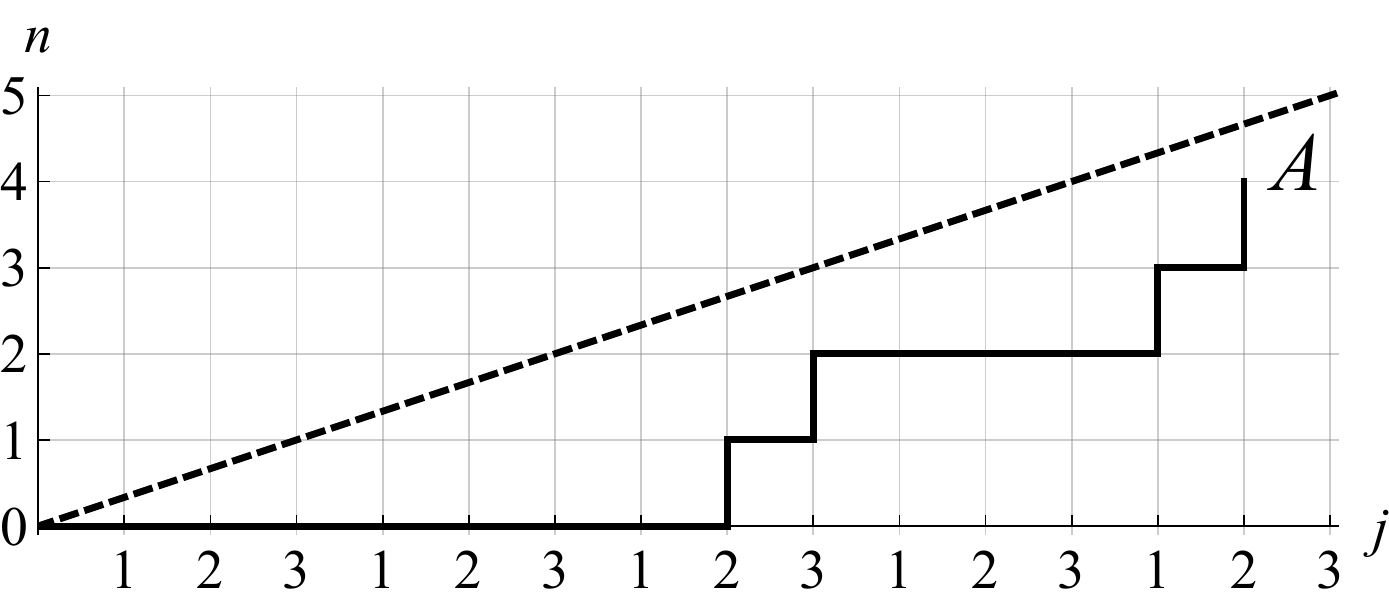}
\caption{There are $C^3_{2,4}=340$ paths to the point $A=(3\cdot4+2,4)$.}
\label{fig:path}
\end{figure}

\begin{proposition}\label{pr:egf}
For given $p$ and $j=1,\dots,p+1$, the EGF of the sequence $C^p_{j,n}$, that is, the series
\begin{equation}\label{egf-pj}
 C^p_{j,0} +C^p_{j,1}t+ C^p_{j,2}\frac{t^2}{2!} +\dots + C^p_{j,n}\frac{t^n}{n!} +\dots
\end{equation}
is equal to
\begin{equation}\label{Fpj} 
 F^p_j(t)= {}_pF_p\left(\frac{j}{p+1},\dots,\widehat1,\dots,\frac{j+p}{p+1};\; 
     \frac{j+1}{p},\dots,\frac{j+p}{p};\;\frac{(p+1)^{p+1}}{p^p}t\right), 
\end{equation}
where $\widehat1$ is the omitted value in the first group of parameters.
\end{proposition}

\begin{proof}
By definition, the series (\ref{pFq}) is characterized by the property that its first term is equal to 1 and the ratio of the coefficients of $z^{n+1}$ and $z^n$ is equal to
\[
 \frac{(a_1+n)\cdots(a_p+n)}{(b_1+n)\cdots(b_q+n)(n+1)}.
\] 
The first term in the series $F^p_j$ is equal to 1. According to the formula (\ref{Cpjn}), the ratio of the coefficients of $t^{n+1}$ and $t^n$ is
\begin{align*}
 \frac{C^p_{j,n+1}}{C^p_{j,n}(n+1)}
  &= \frac{((p+1)n+j)\cdots((p+1)n+j+p)}{(pn+j+1)\cdots(pn+j+p)(n+1)^2}\\
  &= \frac{(\frac{j}{p+1}+n)\cdots(\frac{j+p}{p+1}+n)}
          {(\frac{j+1}{p}+n)\cdots(\frac{j+p}{p}+n)(n+1)^2}\cdot\frac{(p+1)^{p+1}}{p^p}.
\end{align*}
One of the factors in the numerator cancels with $(n+1)$ in the denominator and we arrive to the hypergeometric series (\ref{Fpj}) of the scaled variable $z=(p+1)^{p+1}/p^pt$.
\end{proof}

\begin{remark}
We emphasize that although the numbers $C^p_{j,n}$ and functions $F^p_j$, $f_j$ are defined for all $j>0$, Conjecture \ref{conj:fCpjn} and Proposition \ref{pr:egf} concern only $j$ from 1 to $p+1$. For $j>p+1$, the functions (\ref{egf-pj}), $F^p_j$ and $f_j$ are pairwise distinct. For $j=p+1$, their coincidence follows from the coincidence for $j=1$. Indeed, the defintion (\ref{Cpjn}) implies the easily verifiable identity
\[
 C^p_{p+1,n}=C^p_{1,n+1},
\]
that is, the sequence for $j=p+1$ is obtained from the sequence for $j=1$ by dropping the first term (this is why these numbers are not presented in Table \ref{t:Cpjn}) and, on the other hand, the rules of differentiation of hypergeometric functions and equations (\ref{fL}) imply the equalities
\[
 F^p_{p+1}=(F^p_1)',\quad f_{p+1}=f'_1.
\]
However, we will see that it is convenient to append $f_{p+1}$ to the main set of $f_1,\dots,f_p$.
\end{remark}

\subsection{Narayana polynomials}

The unit step is not the only `magic' example that leads to interesting answers. Let us return to the case $p=1$ and repeat the calculations for the alternating initial data
\begin{equation}\label{stepq}
 u_j=0,~~ j\le0,\quad u_1=u_3=\dots=1,\quad u_2=u_4=\dots=q.
\end{equation}
For the variables $v_j$ (such that $u_j=v_{j+1}v_j$), the corresponding initial data are
\[
 v_0=0,~~ v_1=v_2=1,~~ v_3=q,~~ v_4=q^{-1},~~ v_5=q^2,~~ v_5=q^{-2},~\dotso.
\]
The calculation of derivatives of $f_1=v_1$ by virtue of equation $v'_j=v^2_j(v_{j+1}-v_{j-1})$ brings to
\begin{gather*}
 f_1(0)=1,\quad f'_1(0)=1,\quad f''_1(0)=1+q,\quad f'''_1(0)=1+3q+q^2,\\
 f^{(4)}_1(0)=1+6q+6q^2+q^3,\quad f^{(5)}_1(0)=1+10q+20q^2+10q^3+q^4,~~\dotso.
\end{gather*}
The triangle of the polynomial coefficients is easily identified \cite[A001263]{OEIS} and we find that these are the so-called Narayana polynomials
\begin{equation}\label{Nq}
 N_0(q)=1,\quad N_n(q)= \sum^n_{k=1}\frac{1}{n}\binom{n}{k-1}\binom{n}{k}q^{k-1},\quad n>0,
\end{equation}
which determine a deformation of the Catalan numbers and have applications in problems of counting paths on a lattice, see e.g.~\cite{Bonin_Shapiro_Simion_1993, Sulanke_2000,Lassalle_2012} (in \cite{Bonin_Shapiro_Simion_1993}, this name was used for the polynomials $(1+q)N_n(1+q)$). The values $N_n(q)$ for integer $q$ give well-known combinatorial sequences, in particular, the Catalan numbers correspond to $q=1$, for $q=2$ we obtain the sequence A001003 (little Schr\"oder numbers) and so on (the OEIS includes the examples up to $q=11$). Thus, if the conjecture that $f^{(n)}_1(0)=N_n(q)$ for all $n$ is true, then $f_1$ is the EGF for the sequence $N_n(q)$. We will prove this in Sect.~\ref{s:p1q}.

%-------------------------------------------------------------------------------
\section{Determinant formulae}\label{s:det}

We derive formulae for solving the recurrence relations (\ref{fL})
\begin{equation}\label{fL'}
 f_{-p}=\dots=f_0=1,\quad f_{j-1}f'_j-f'_{j-1}f_j=f_{j-p-1}f_{j+p},\quad j=1,2,\dots,
\end{equation}
with given functions $f_1,\dots,f_p$ (at the moment, their normalization given in (\ref{fic}) is not necessary). Assume that all variables in (\ref{fL'}) do not vanish identically, so that the division by $f_{j-p-1}$ is possible. It turns out that then there is a complete cancellation and $f_{j+p}$ is expressed by a differential polynomial in the initial functions. More precisely, the answer is written in terms of the Wronskians, for which we use the notation $W(y_0,\dots,y_m)=\det\bigl(y^{(i)}_j\bigr)\big|^m_{i,j=0}$. Let us explain their structure by examples. For $p=1$, the Hankel determinants (\ref{f1}) appear. For $p=2$, the first three variables are
\[
 f_1,~ f_2,~ f_3=f'_1,
\]
the next three variables are expressed by the second order Wronskians
\[
 f_4=\begin{vmatrix} 
  f_1  & f_2\\ 
  f'_1 & f'_2\end{vmatrix},\quad
 f_5=\begin{vmatrix} 
  f_2  & f'_1\\ 
  f'_2 & f''_1\end{vmatrix},\quad
 f_6=\begin{vmatrix} 
  f'_1  & f'_2\\ 
  f''_1 & f''_2\end{vmatrix},
\]
the calculation of the next three variables gives the third order Wronskians
\[
 f_7=\begin{vmatrix}
  f_1   & f_2   & f'_1\\ 
  f'_1  & f'_2  & f''_1\\
  f''_1 & f''_2 & f'''_1
  \end{vmatrix},\quad
 f_8=\begin{vmatrix} 
  f_2   & f'_1   & f'_2\\ 
  f'_2  & f''_1  & f''_2\\
  f''_2 & f'''_1 & f'''_2
  \end{vmatrix},~~
 f_9=\begin{vmatrix} 
  f'_1   & f'_2   & f''_1\\ 
  f''_1  & f''_2  & f'''_1\\
  f'''_1 & f'''_2 & f^{(4)}_1
  \end{vmatrix},
\]
and so on. In the general case, we write down the interlacing sequence of derivatives
\begin{equation}\label{fff}
 f_1,\dots,f_p,f'_1,\dots,f'_p,f''_1,\dots,f''_p,f'''_1,\dots,f'''_p,\dots
\end{equation}
and construct the groups of Wronskians of the same size for the segments of this sequence, starting in turn with the elements $f_1,\dots,f_p,f'_1$. In each Wronskian, all rows are segments of the sequence (\ref{fff}) shifted by $p$ positions. This leads to equations
\begin{gather}
\label{gj}
 g_{i+pn}=f^{(n)}_i,\quad i=1,\dots,p,\quad n=0,1,2,\dots,\\
\label{fj}
 f_{k+(p+1)m}=W(g_k,\dots,g_{k+m}),\quad k=1,\dots,p+1,\quad m=0,1,2,\dots,
\end{gather}
which, if desired, can be combined into
\begin{equation}\label{fjshort}
 f_{k+(p+1)m} = \det\Bigl(f^{(\lfloor\frac{k+i-1}{p}\rfloor+j)}_{1+(k+i-1)\bmod p}\Bigr)\Big|^m_{i,j=0},
 \quad k=1,\dots,p+1,~~ m=0,1,\dotso.
\end{equation}
Similar representations for solutions of integrable equations are widely known, see e.g.~papers \cite{Leznov_1980, Leznov_Savel'ev_Smirnov_1981} devoted to the Toda lattices and the monograph \cite{Vein_Dale_1999} containing a lot of examples. The proof of formulae (\ref{gj}) and (\ref{fj}) given below is based on the Jacobi identity for Wronskians
\begin{equation}\label{wwwwww}
 W(f,G)(W(G,h))'-(W(f,G))'W(G,h)=W(G)W(f,G,h),
\end{equation}
where $f$ and $h$ are arbitrary functions and $G$ is any tuple of functions. 

\begin{proposition}\label{pr:detf}
Let functions $f_1,\dots,f_p$ be infinitely differentiable then formulae (\ref{gj}) and (\ref{fj}) define a solution of equation (\ref{fL'}).
\end{proposition}

\begin{proof}
For $m=0$ and $k=1,\dots,p$, the equality (\ref{fj}) is satisfied identically. For $m=0$ and $k=p+1$, we obtain $f_{p+1}=f'_1$, which is true due to (\ref{fL'}). Now let $m>0$ and $j=k+(p+1)m$, where $k=1,\dots,p+1$. Let $G$ denote the tuple $g_k,\dots,g_{k+m-1}$. For $k\ne1$, we have
\begin{align*}
 & f_j      = W(g_k,\dots,g_{k+m}) = W(G,g_{k+m}),\\
 & f_{j-1}  = W(g_{k-1},\dots,g_{k+m-1}) = W(g_{k-1},G),\\
 & f_{j-p-1}= W(g_k,\dots,g_{k+m-1}) = W(G),\\
 & f_{j+p}  = W(g_{k-1},\dots,g_{k+m}) = W(g_{k-1},G,g_{k+m}),
\end{align*}
and if $k=1$ then
\begin{align*}
 & f_j      = W(g_1,\dots,g_{1+m}) = W(G,g_{1+m}),\\
 & f_{j-1}  = W(g_{p+1},\dots,g_{p+m}) = W(g'_1,\dots,g'_m) = W(1,G),\\
 & f_{j-p-1}= W(g_1,\dots,g_m) = W(G),\\
 & f_{j+p}  = W(g_{p+1},\dots,g_{p+1+m}) = W(g'_1,\dots,g'_{1+m}) = W(1,G,g_{1+m}),
\end{align*}
where we use the obvious identity $W(G')=W(1,G)$. In both cases the Jacobi identity proves the equality $f_{j-1}f'_j-f'_{j-1}f_j=f_{j-p-1}f_{j+p}$, as required.
\end{proof}

\begin{remark}
In the proof, the division was not used, hence equations (\ref{gj}) and (\ref{fj}) give a solution of (\ref{fL'}) even when some determinants are equal to zero. This only leads to non-uniqueness of the solution. A rather meaningful example, when all determinants with sufficiently large numbers turn to 0 and the lattice is restricted onto a finite segment, is obtained when $f_1,\dots,f_p$ are arbitrary polynomials in $t$. However, in what follows we require that all $f_j\ne0$. In this case the solution is unique, which is obvious from the fact that equations (\ref{fL'}) are uniquely solvable with respect to $f_{j+p}$. 
\end{remark}

A direct use of equations (\ref{gj}) and (\ref{fj}) is to represent the general solution of the lattice equation in terms of arbitrary functions $f_1,\dots,f_p$. If these functions have such a remarkable property that some explicit formula for the Wronskians constructed from them is known for $t=0$ then according to (\ref{uff}) one can explicitly calculate the values
\begin{equation}\label{uff0}
 u_j(0)=\frac{f_{j-p-1}(0)f_{j+p}(0)}{f_{j-1}(0)f_j(0)},\quad j=1,2,\dots
\end{equation}
(assuming that $f_{-p}=\dots=f_0=1$) and, by construction, we know the solution which correspond to these special initial data. 

Conversely, solving equations (\ref{uff0}) with respect to $f_j(0)$ gives expressions for the Wronskians at $t=0$ in terms of the initial data $u_j(0)$. This leads to the following formula, which is easy to prove by induction.

\begin{proposition}\label{pr:fuu0}
Solution of equation (\ref{uff0}) with given $u_j(0)\ne0$ and the initial condition $f_{-p}(0)=\dots=f_p(0)=1$ is of the form
\begin{equation}\label{fuu0}
 f_j(0)= \prod^j_{i=1}u_i(0)^{\lfloor\frac{j-i}{p}\rfloor - \lfloor\frac{j-i}{p+1}\rfloor},\quad j=1,2,\dotso.
\end{equation}
\end{proposition}

\begin{remark}
The occurrences of each factor $u_i(0)$ in the products, starting from $f_{i+p}(0)$, are governed by the sequence of powers $\lfloor\frac{k+p}{p}\rfloor -\lfloor\frac{k+p}{p+1}\rfloor$, $k=0,1,2,\dotso$. This family of sequences has many applications, see A004526, A008615, A008679, A165190, etc.~in the OEIS.
\end{remark}

If we know the functions $f_1,\dots,f_p$ that correspond to the given initial data $u_j(0)$ then equation (\ref{fuu0}) gives a method of calculation of Casoratians constructed from the coefficients of the Taylor expansions
\begin{equation}\label{faaa}
 f_j=a_{j,0}+a_{j,1}t+a_{j,2}\frac{t^2}{2!}+\dots+a_{j,n}\frac{t^n}{n!}+\dots,\quad j=1,\dots,p.
\end{equation}
More precisely, the Wronskians (\ref{fj}) at $t=0$ turn into Casoratians constructed along the same recipe as before, with the sequence (\ref{fff}) replaced by
\begin{equation}\label{aaa}
 a_{1,0},\dots,a_{p,0},a_{1,1},\dots,a_{p,1},a_{1,2},\dots,a_{p,2},a_{1,3},\dots,a_{p,3},\dots,
\end{equation}
and equations (\ref{gj}) and (\ref{fj}) replaced by
\begin{gather}
\label{bj}
 b_{i+pn}=a_{i,n},\quad i=1,\dots,p,\quad n=0,1,2,\dots,\\
\label{fj0}
 f_{k+(p+1)m}(0)=\det(b_{k+pi+j})\big|^m_{i,j=0},\quad k=1,\dots,p+1,\quad m=0,1,2,\dotso.
\end{gather}
For instance, if $p=2$ then the Wronskians listed at the beginning of the section turn into
\begin{gather*}
 f_1(0)=a_{1,0},\quad f_2(0)=a_{2,0},\quad f_3(0)=a_{1,1},\\
 f_4(0)=\begin{vmatrix} 
  a_{1,0} & a_{2,0}\\ 
  a_{1,1} & a_{2,1}\end{vmatrix},\quad
 f_5(0)=\begin{vmatrix} 
  a_{2,0} & a_{1,1}\\ 
  a_{2,1} & a_{1,2}\end{vmatrix},\quad
 f_6(0)=\begin{vmatrix} 
  a_{1,1} & a_{2,1}\\ 
  a_{1,2} & a_{2,2}\end{vmatrix},\\
 f_7(0)=\begin{vmatrix}
  a_{1,0} & a_{2,0} & a_{1,1}\\ 
  a_{1,1} & a_{2,1} & a_{1,2}\\
  a_{1,2} & a_{2,2} & a_{1,3}
  \end{vmatrix},\quad
 f_8(0)=\begin{vmatrix} 
  a_{2,0} & a_{1,1} & a_{2,1}\\ 
  a_{2,1} & a_{1,2} & a_{2,2}\\
  a_{2,2} & a_{1,3} & a_{2,3}
  \end{vmatrix},~~
 f_9(0)=\begin{vmatrix} 
  a_{1,1} & a_{2,1} & a_{1,2}\\ 
  a_{1,2} & a_{2,2} & a_{1,3}\\
  a_{1,3} & a_{2,3} & a_{1,4}
  \end{vmatrix}.
\end{gather*}
Notice that the bottom right entries of determinants $f_j(0)$ form exactly the sequence (\ref{aaa}), and the complementary minors are determinants $f_{j-p}(0)$ of smaller size from the same sequence of determinants. Consequently, if $f_j(0)\ne0$ for all $j$ then the transformation between sequences (\ref{aaa}) and (\ref{fj0}) is one-to-one. This transformation is especially helpful in the situation when the sequence (\ref{aaa}) is such that, firstly, the corresponding series (\ref{faaa}) are some known functions and, secondly, the determinants (\ref{fj0}) are expressed in a closed form. This is exactly the case with the generalized Catalan numbers. The corresponding Taylor series are given in Proposition \ref{pr:egf} and all determinants of the form (\ref{fj0}) are just equal to 1. This property is well-known in combinatorics, especially for the case $p=1$.

\begin{proposition}[\cite{Stanley_1999,Aigner_1999,Layman_2001}]\label{pr:Hankel}
The Catalan numbers $c_n=C^1_{1,n}$ satisfy the equalities (and are uniquely defined by them)
\begin{equation}\label{cid}
 \left|\begin{matrix} 
   c_0 & \dots & c_k\\
  \hdotsfor[2]{3} \\
   c_k & \dots & c_{2k}
 \end{matrix}\right|
 =\left|\begin{matrix} 
   c_1 & \dots & c_{k+1}\\
   \hdotsfor[2]{3} \\
   c_{k+1} & \dots & c_{2k+1}
 \end{matrix}\right|= 1,\quad k=0,1,2,\dotso.
\end{equation}
\end{proposition}

\begin{proposition}[\cite{Krattenthaler_2010}]\label{pr:gen-Hankel}
The generalized Catalan numbers $C^p_{j,n}$ satisfy the equalities (and are uniquely defined by them)
\begin{equation}\label{Cid}
 \det\left(C^p_{1+(i+k-1)\bmod p,\,\lfloor\frac{i+k-1}{p}\rfloor+j}\right)\Big|^m_{i,j=0}=1,
\end{equation}
for $k=1,\dots,p+1$, $m=0,1,2,\dots$.
\end{proposition}

Using these identities, we can prove Conjecture \ref{conj:fCpjn} from which, taking into account Proposition \ref{pr:egf} on the EGFs for the generalized Catalan numbers, Proposition \ref{pr:step} follows as a corollary.

\begin{proof}[{\bfseries Proof of Conjecture \ref*{conj:fCpjn}}]
By solving equation (\ref{uff0}) with $u_j(0)=1$, we find that $f_j(0)=1$ for all $j$. On the other hand $f_j(0)$ are equal to determinants (\ref{fj0}) constructed from the Taylor coefficients (\ref{faaa}). The identity (\ref{Cid}) implies that these coefficients coincide with the generalized Catalan numbers.
\end{proof}

The combinatorial proof of Proposition \ref{pr:gen-Hankel} is given in \cite{Krattenthaler_2010}. It is based on the technique of counting paths on the 2D lattice. In fact, \cite[Theorem 6]{Krattenthaler_2010} contains a more general identity with additional integer parameters $\alpha_0,\dots,\alpha_m$; the identity (\ref{Cid}) is its particular case for $\alpha_i=i$. For the case $p=1$, an elementary proof (using the identity $c_{n+1}=c_0c_n+\cdots+c_nc_0$) is given in \cite{Adler_Shabat_2018}. 

We do not present these proofs; instead, the rest of the paper describes the reduction method, which allows us to obtain an alternative proof of Proposition \ref{pr:step} and thereby Conjecture \ref{conj:fCpjn}. With this approach, the determinant relations are not needed at all, but we can use them to obtain an independent proof of the identities (\ref{cid}) and (\ref{Cid}). Moreover, we obtain functions $f_1,\dots,f_p$ not only for the unit step, but also for some more general $2p$-parametric family of initial data $u_j(0)$. Equation (\ref{fuu0}) gives a computation method for determinants constructed from the Taylor coefficients of these functions. 

In the example with the Narayana polynomials, solving of equation (\ref{uff0}) for $p=1$ and with the initial data $u_j(0)$ of the form (\ref{stepq}) brings to the following generalization of the identities (\ref{cid}). Our experimental fact, that $f_1$ serves as the EGF for $N_n(q)$, follows from the validity of these identities and vice versa.

\begin{proposition}[\cite{Petkovic_Barry_Rajkovic_2012}]\label{pr:Hankelq}
The polynomials $N_n(q)$ (\ref{Nq}) satisfy the equalities (and are uniquely defined by them)
\begin{equation}\label{cidq}
 \left|\begin{matrix} 
   N_0(q) & \dots & N_k(q)\\
  \hdotsfor[2]{3} \\
   N_k(q) & \dots & N_{2k}(q)
 \end{matrix}\right|
 =\left|\begin{matrix} 
   N_1(q) & \dots & N_{k+1}(q)\\
  \hdotsfor[2]{3} \\
   N_{k+1}(q) & \dots & N_{2k+1}(q)
 \end{matrix}\right|= q^{\frac{1}{2}k(k+1)},\quad k=0,1,2,\dotso.
\end{equation}
\end{proposition}

%-------------------------------------------------------------------------------
\section{Symmetry reductions for $p=1$}\label{s:p1}

An evolution partial differential or lattice equation can be turned into a finite-dimensional system by imposing some constraint consistent with the dynamics. The standard source of such reductions are the stationary equations of generalized symmetries. Indeed, let $u_t=f[u]$ be the given equation and let $u_\tau=g[u]$ be its symmetry, that is
\[
 [\partial_t,\partial_\tau]=0 \quad\Leftrightarrow\quad \partial_t(g)=f_*(g),
\]
where $f_*$ is the linearization operator. Then the stationary equation $g=0$ defines an invariant manifold, since $\partial_t(g)=0|_{g=0}$. Let us consider the Volterra lattice (that is, BL$_1$) 
\begin{equation}\label{VL}
 u'_j=u_j(u_{j+1}-u_{j-1})
\end{equation} 
as an example, keeping the prime notation for the main derivation $\partial_t$. To build a constraint, we use the higher symmetry, the classical scaling symmetry and the master-symmetry \cite{Zhang_Tu_Oevel_Fuchssteiner_1991, Cherdantsev_Yamilov_1995}:
\begin{align}
\label{t2}
 & u_{j,t_2}=u_j(T-T^{-1})\bigl(u_j(u_{j+1}+u_j+u_{j-1})\bigr),\\
\label{tau0}
 & u_{j,\tau_0}=u_j,\\
\label{tau1}
 & u_{j,\tau_1}= u_j((j+3)u_{j+1}+u_j-ju_{j-1}).
\end{align}
The corresponding derivations satisfy the relations
\begin{equation}\label{commutators}
 [\partial_{t_2},\partial_t]=0,\quad 
 [\partial_{\tau_0},\partial_t]=\partial_t,\quad 
 [\partial_{\tau_1},\partial_t]=\partial_{t_2},
\end{equation}
which imply that $\partial_t$ commutes with $\partial_{t_2}$, $t\partial_t+\partial_{\tau_0}$ and $t\partial_{t_2}+\partial_{\tau_1}$. The general form of the symmetry of order not greater than 2 with respect to the shifts is given by the linear combination
\begin{equation}\label{sym}
 c_1(t\partial_{t_2}+\partial_{\tau_1}) +c_2\partial_{t_2} +c_3(t\partial_t+\partial_{\tau_0}) +c_4\partial_t.
\end{equation}
We set $c_1=1$ and $c_2=0$ (by use of the change $t\to t-t_0$) and denote $c_3=-4a$ and $c_4=b-2$, which slightly simplifies further relations. Then the stationary equation for the symmetry (\ref{sym}) gives, after dividing by $u_j$, the 5-point difference equation
\begin{multline}\label{cA}\qquad
 A_j=tu_{j+1}(u_{j+2}+u_{j+1}+u_j-4a)-tu_{j-1}(u_j+u_{j-1}+u_{j-2}-4a)\\
   +(j+b+1)u_{j+1}+u_j-(j+b-2)u_{j-1}-4a=0,\qquad
\end{multline}
where $a$ and $b$ are arbitrary constants. We can forget about its origin from symmetries and directly verify its compatibility with the lattice equation by checking that differentiation by virtue of (\ref{VL}) satisfies the identity
\[
 A'_j= u_{j+1}A_{j+1}-u_{j-1}A_{j-1}.
\]
From here it follows that if a solution of equation (\ref{VL}) satisfies the constraints $A_j=0$ for all $j$ at $t=0$ then this is also true for all $t$ for which this solution is defined. 

Equation (\ref{cA}) admits a reduction of order by two. First, we notice that it can be written in the form $A_j=B_j+B_{j+1}=0$, which gives
\begin{multline}\label{cB}\qquad
 B_j=u_{j+1}(tu_{j+1}+tu_{j+2}+j+b+1)-u_j(tu_{j-1}+tu_j+j+b-1)\\
  +2a(2tu_j-2tu_{j+1}-1) -4(-1)^jc=0.\qquad
\end{multline}
Next, multiplying by the integrating factor $tu_j+tu_{j+1}+j+b$ brings this to the form $C_{j+1}-C_j=0$ and the equality $C_j=0$ gives the 3-point equation
\begin{multline}\label{cC}\qquad
 (tu_{j-1}+tu_j+j+b-1)(tu_j+tu_{j+1}+j+b)u_j\\
     = a(2tu_j+j+b-\tfrac{1}{2})^2 -(-1)^jc(4tu_j+2j+2b-1) +d.\qquad
\end{multline}
One can verify that the compatibility with (\ref{VL}) is preserved under these manipulations, provided that the integration constants $c$ and $d$ do not depend on $t$. For example, for the constraint (\ref{cB}) this follows from the relation
\begin{equation}\label{cB'}
 B'_j=u_{j+1}(B_{j+1}+B_j)-u_j(B_j+B_{j-1}).
\end{equation}

\begin{remark}
Equation (\ref{cC}) is equivalent, up to notations, to the discrete Painlev\'e equation dP$_{34}$ \cite{Grammaticos_Ramani_2014}. In \cite{Adler_Shabat_2019}, it was shown that the constraint (\ref{cC}) turns the lattice equation (\ref{VL}) into the closed system for any pair of variables $u_{j-1},u_j$, which is equivalent to equation P$_5$ if $a\ne0$ or P$_3$ if $a=0$. Similar reductions for matrix Volterra lattices were studied in \cite{Adler_2020}. The choice of $c_1=0$ and $c_2=1$ in (\ref{sym}) leads to a simpler reduction in terms of the Painlev\'e equations dP$_1$ and P$_4$ \cite{Fokas_Its_Kitaev_1991}. However, in this paper we do not use these results and restrict ourselves to studying special solutions of the Painlev\'e equations that arise under the additional boundary condition that terminates the lattice on the right half-line. In this case, a simplification occurs, the equations are linearized and the answer is expressed in terms of hypergeometric functions.
\end{remark}

A nice feature of our constraint is that for $t=0$ it simplifies to the explicit formula (here we assume that the denominator is not zero)
\begin{equation}\label{icp1}
 u_j(0)=\frac{a(j+b-1/2)^2-(-1)^jc(2j+2b-1)+d}{(j+b-1)(j+b)},
\end{equation}
which defines the family of initial data for solutions governed by this constraint. More precisely, this formula covers only those solutions for which all $u_j(t)$ are regular at $t=0$, but these are the only ones we need. Under the termination on the half-line, the parameter $d$ is fixed from the condition $u_0(0)=0$.

When solving the problem on the half-line, it suffices to use the constraint equations in the form (\ref{cB}) for $j\ge0$: it is easy to see from (\ref{cB'}) that equations
\[
 u_0=0,\quad B_0=B_1=B_2=\dots=0
\]
also define an invariant solution submanifold for (\ref{VL}). Moreover, to find $u_1$ only one equation $B_0=0$ is needed: all other equations are consequences of it and are satisfied automatically, as shown by equation (\ref{cB'}) which plays the role of a mechanism that provides the induction step. To prove the following proposition we use the standard hypergeometric equations
\begin{equation}\label{hyper}
 tf''=(kt-b)f'+akf,\qquad tf''=-af'+kf,
\end{equation}
which define, respectively, functions $f={}_1F_1(a,b,kt)$ and $f={}_0F_1(a,kt)$ as the unique regular solutions satisfying the normalization $f(0)=1$.

\begin{proposition}\label{pr:p1}
Let $u_j$ satisfy the lattice equation $u'_j=u_j(u_{j+1}-u_{j-1})$ for $j>0$, the boundary condition $u_0=0$ and the initial data
\begin{equation}\label{icp10}
 u_1(0)=\frac{2a+4c}{b+1},\quad u_j(0)=\frac{aj(j+2b-1)-(-1)^jc(2j+2b-1)+c(2b-1)}{(j+b-1)(j+b)},
\end{equation}
where $b\ne-1,-2,\dotso$. Then $u_1=f'_1/f_1$, where
\begin{alignat}{2}
\label{f11}
 f_1&= {}_1F_1\Bigl(\frac{1}{2}+\frac{c}{a},b+1,4at\Bigr)&& \text{for}~ a\ne0,\\
\label{f01}
 f_1&= {}_0F_1(b+1,4ct)& \quad& \text{for}~ a=0.
\end{alignat}
\end{proposition}

\begin{proof}
We substitute $u_0(0)=0$ into (\ref{icp1}) and obtain $d=c(2b-1)-a(b-1/2)^2$, which gives the initial data (\ref{icp10}). In the formula for $u_1(0)$, the factors $b$ in the numerator and denominator cancel, so that the value $b=0$ is admissible, which is easy to justify by analysing the constraint in the form (\ref{cB}). Therefore the initial data (\ref{icp10}) satisfy the equations of the constraint (\ref{cB}) for $j\ge0$ and this implies that this constraint is fulfilled for all $t$. Equation $B_0=0$ and the boundary condition $u_0=0$ lead to the closed system for the variables $u_1$ and $u_2$:
\[
 u'_1=u_1u_2,\quad u_1(tu_1+tu_2+b+1)-2a(2tu_1+1)-4c=0,
\]
which is equivalent to the Riccati equation for $u_1$. The usual linearization $u_1=f'_1/f_1$ (it coincide with the substitution from Proposition \ref{pr:uf}) brings to
\begin{equation}\label{f1''}
 f''_1=\Bigl(4a-\frac{b+1}{t}\Bigr)f'_1+\frac{2a+4c}{t}f_1,\quad f_1(0)=1, 
\end{equation}
which defines, according to (\ref{hyper}), functions (\ref{f11}) or (\ref{f01}).
\end{proof}

We remark that if $a\ne0$ then the formula (\ref{icp10}) can give $u_j(0)=0$ for some $j>0$, due to a special choice of parameters $b$ and $c$. Then the solution constructed from $f_1$ is restricted to a finite segment of the lattice. Nevertheless, the formula (\ref{f11}) for $f_1$ remains true and such cases correspond to the degeneration of the hypergeometric function into a polynomial.

Under the choice $a=1$, $b=1$ and $c=0$, the initial data (\ref{icp10}) turn into the unit step and (\ref{f11}) gives the answer (\ref{egf-c}) which defines the EGF for the Catalan numbers. This proves Conjecture \ref{conj:fCpjn} for $p=1$ and, as a corollary, the identity (\ref{cid}).

%-------------------------------------------------------------------------------
\section{A generalization for the Toda lattice}\label{s:p1q}

The constraint (\ref{cA}) admits a generalization allowing, in particular, to obtain the solution for initial data (\ref{stepq}) corresponding to Narayana polynomials. It is more convenient to write this generalization in the variables
\begin{equation}\label{uuyz}
 y_j=u_{2j+2}+u_{2j+1},\quad z_j=u_{2j+1}u_{2j}.
\end{equation}
It is easy to check that if $u_j$ satisfy the Volterra lattice (\ref{VL}) then $y_j$ and $z_j$ satisfy the Toda lattice in the Flaschka form
\begin{equation}\label{TL}
 y'_j=z_{j+1}-z_j,\quad z'_j=z_j(y_j-y_{j-1}).
\end{equation}
The derivations (\ref{t2})--(\ref{tau1}) also can be rewritten in these variables as follows (we use the derivative with respect to $t$ in the first equation only to bring it to a more compact form, it can be expanded in virtue of (\ref{TL})):
\begin{align*}
& y_{j,t_2}=(z_{j+1}+z_j+y^2_j)',\quad  z_{j,t_2}=(z_j(y_j+y_{j-1}))',\\
& y_{j,\tau_0}=y_j,\quad z_{j,\tau_0}=2z_j,\\
& y_{j,\tau_1}=(2j+5)z_{j+1}-(2j+1)z_j+y^2_j,\quad z_{j,\tau_1}=z_j((2j+4)y_j-2jy_{j-1}). 
\end{align*}
The commutator relations (\ref{commutators}) do not change. It is easy to see that a very simple additional symmetry appears in the new variables, corresponding to the shift $y\to y+\varepsilon$:
\[
 y_{j,t_0}=1,\quad z_{j,t_0}=0
\]
(in the variables $u_j$ this symmetry is nonlocal). By adding a term with $\partial_{t_0}$ to the linear combination (\ref{sym}) and passing to the stationary equation we obtain the following constraint:
\begin{equation}\label{cPQ}
\begin{aligned}
 & P_j=tz_{j+1}(y_j+y_{j+1}-4a)-tz_j(y_{j-1}+y_j-4a)\\
 &\qquad\qquad  +(2j+3+b)z_{j+1}-(2j-1+b)z_j+y^2_j-4ay_j+\mu=0,\\
 & Q_j=t(z_{j+1}-z_{j-1}+y^2_j-y^2_{j-1}-4ay_j+4ay_{j-1})\\
 &\qquad\qquad +(2j+2+b)y_j-(2j-2+b)y_{j-1}-8a=0,
\end{aligned}
\end{equation}
where $a$, $b$ and $\mu$ are arbitrary constants. Under the substitutions (\ref{uuyz}) the equalities hold
\[
 P_j=u_{2j+2}A_{2j+2}+u_{2j+1}A_{2j+1}+\mu,\quad Q_j=A_{2j+1}+A_{2j},
\] 
where $A_j$ is the left-hand side of the constraint (\ref{cA}), which means that the constraint (\ref{cPQ}) includes (\ref{cA}) as a 
particular case for $\mu=0$. Similar to the case of the constraint (\ref{cA}), one can forget about the origin of equations (\ref{cPQ}) from symmetries and prove compatibility with the dynamics in $t$ by straightforward checking the relations
\begin{equation}\label{cPQ'}
 P'_j=z_{j+1}Q_{j+1}-z_jQ_j,\quad Q'_j=P_j-P_{j-1}.
\end{equation}
The mapping $(y_{j-1},z_{j-1},y_j,z_j)\mapsto(y_j,z_j,y_{j+1},z_{j+1})$ defined by equations (\ref{cPQ}) is of fourth order with respect to the shift. It has two first integrals, which can be extracted from the matrix Lax representation for the stationary symmetry equation
\begin{equation}\label{LM}
 2(\lambda^2-4a\lambda+\mu)\partial_\lambda(L_j)=M_{j+1}L_j-L_jM_j,
\end{equation}
where $\lambda$ is the spectral parameter. Equations (\ref{cPQ}) are equivalent to (\ref{LM}) under the choice
\[
 L_j=z^{-1/2}_j\begin{pmatrix} 0 & 1\\ -z_j & \lambda-y_j \end{pmatrix},\quad
 M_j=\begin{pmatrix}
  -s_j & 2(r_{j-1}+1)\\ 
  -2z_j(r_j+1) & s_j
 \end{pmatrix},
\]
where
\[
 r_j=t(y_j+\lambda-4a)+2j+b,\quad s_j=t(z_j-z_{j-1})-(y_{j-1}-\lambda)r_{j-1}+2\lambda-4a.
\]
Equation (\ref{LM}) is polynomial in $\lambda$ and it implies the realtion
\[
 \det M_{j+1}=\det M_j \mod(\lambda^2-4a\lambda+\mu),
\]
which provides two first integrals after elimination of powers of $\lambda$ greater than one. We do not write them out explicitly, since they are quite cumbersome. Apparently, by lowering the order with their help and making additional changes of variables, it is possible to reduce the constraint under consideration to one of the Painlev\'e equations, but this is a separate exercise that is beyond the scope of our work. We use these first integrals only at $t=0$ to determine the initial data served by our constraint. Calculating $\det M_j$ for $t=0$, replacing $\lambda^2$ with $4a\lambda-\mu$, and collecting the coefficients at $\lambda^0$ and $\lambda^1$, we arrive at equations
\begin{equation}\label{integrals}
\begin{gathered}
 (2j+b-2)(2j+b)(y_{j-1}(0)-2a)=\alpha,\\
 4(2j+b-1)(2j+b+1)z_j(0) - ((2j+b-2)y_{j-1}(0)+4a)^2 +\mu(2j+b)^2=\beta,
\end{gathered}
\end{equation}
where parameters $\alpha$ and $\beta$ are the values of first integrals. These equations determine the initial data for the regular solutions of the lattice equation (\ref{TL}) satisfying the constraint (\ref{cPQ}) (cf. (\ref{icp1})).

As in the previous section, we restrict ourselves to the study of equations on the half-line, which is the linearizable case. The corresponding reduction defining an invariant submanifold on solutions of (\ref{TL}) is given by equations
\[
 z_0=0,\quad P_0=P_1=\dots=0,\quad Q_1=Q_2=\dots=0.
\]
To obtain a closed system we need just one equation $P_0=0$ since all other equations of the constraint are its consequences, according to (\ref{cPQ'}).

\begin{proposition}\label{pr:r1q}
Let the variables $y_j$ and $z_j$, $j\ge0$, satisfy the lattice equation (\ref{TL}) with the boundary conditions $z_0=0$ and initial data
\begin{equation}\label{icp1q0}
 y_j(0)=2a+\frac{2bd}{(2j+b)(2j+b+2)},\quad 
 z_j(0)=\frac{4j(j+b)(c^2(2j+b)^2-d^2)}{(2j+b-1)(2j+b)^2(2j+b+1)},
\end{equation}
where $b\ne0,-1,-2,\dotso$. Then $y_0=g'/g$, where
\begin{alignat}{2}
\label{g1}
 g&= e^{2(a-c)t}{}_1F_1\Bigl(1+\frac{b}{2}+\frac{d}{2c},b+2,4ct\Bigr)&& \text{for}~ c\ne0,\\
\label{g0}
 g&= e^{2at}{}_0F_1(b+2,2dt)& \quad& \text{for}~ c=0.
\end{alignat}
\end{proposition}

\begin{proof}
The initial data are determined by equations (\ref{integrals}), with parameter $\beta$ fixed by condition $z_0=0$. Equations (\ref{icp1q0}) are obtained by denoting $\mu=4a^2-4c^2$ and $\alpha=2bd$. It is also possible to check directly that these formulae give a solution of equations (\ref{cPQ}) for $t=0$. This implies that a solution with such initial data satisfies these equations for all $t$. Equation $P_0=0$ reduces the lattice equations to the closed system for the variables $y_0,y_1$ and $z_1$: 
\[
 y'_0=z_1,\quad z'_1=z_1(y_1-y_0),\quad tz_1(y_0+y_1-4a)+(b+3)z_1+(y_0-2a)^2-4c^2=0.
\]
This is equivalent to one equation
\[
 ty''_0 +2t(y_0-2a)y'_0 +(b+3)y'_0 +(y_0-2a)^2-4c^2=0
\]
which admits the first integral with the value fixed by initial condition for $y_0$:
\[
 ty'_0+t(y_0-2a)^2+(b+2)(y_0-2a)-4c^2t-2d=0.
\]
The substitution $y_0=g'/g$ brings to a linear equation equivalent to (\ref{hyper}) up to multiplying $g$ by a suitable exponent, which gives the answers.
\end{proof}

Under the substitution (\ref{uuyz}), solutions of the Volterra lattice are recovered from the Toda lattice solutions up to one arbitrary parameter. Comparison with formulae from Proposition \ref{pr:uf} gives
\[
 y_0=\frac{g'}{g}=u_2+u_1=\frac{f'_2}{f_2}=\frac{f''_1}{f'_1},
\]
whence it follows (taking the condition $f_1(0)=1$ into account)
\begin{equation}\label{fg}
 f_1=1+\kappa\int^t_0g(x)dx,
\end{equation}
where $\kappa$ is an arbitrary constant. 

We remark that instead of (\ref{uuyz}) one can also use the substitution $y_j=u_{2j+1}+u_{2j}$, $z_j=u_{2j}u_{2j-1}$ which also relates the lattice equations (\ref{TL}) and (\ref{VL}). For this substitution the calculation is even simpler and we obtain just $f_1=g$ instead of (\ref{fg}). However, two families of solutions constructed in this way do not coincide.

In particular, to construct the solution with initial data (\ref{stepq}) leading to the Narayana polynomials, we use the substitution (\ref{uuyz}), which gives
\[
 y_0(0)=y_1(0)=\dots=1+q,\quad z_1(0)=z_2(0)=\dots=q.
\]
These values belong to the family (\ref{icp1q0}) for $2a=1+q$, $b=1$, $c^2=q$ and $d=0$; moreover, we have $u_1(0)=f'_1(0)/f_1(0)=\kappa g(0)=\kappa=1$. As the result, we obtain the following answer.

\begin{proposition}\label{pr:stepq}
The solution of the Volterra lattice (\ref{VL}) with the initial data $u_1=u_3=\dots=1$, $u_2=u_4=\dots=q$ at $t=0$ has the first component $u_1=f'_1/f_1$ with
\[
 f_1= 1+\int^t_0{}_1F_1\bigl(\tfrac{3}{2},3,4cx\bigr)e^{(c-1)^2x}\,dx,\quad c=\sqrt{q}>0.
\]
\end{proposition}

If $c=1$ then this formula is simplified to $f_1={}_1F_1(\tfrac{1}{2},2,4t)$ which correspond to the Catalan numbers, as we already know. For general $c$, the Taylor expansion of $f_1$ brings to the coefficients (\ref{Nq}) as expected; hence, we obtain the identity (\ref{cidq}) as a corollary.

\begin{figure}[t]
\centering
\includegraphics[width=0.45\textwidth]{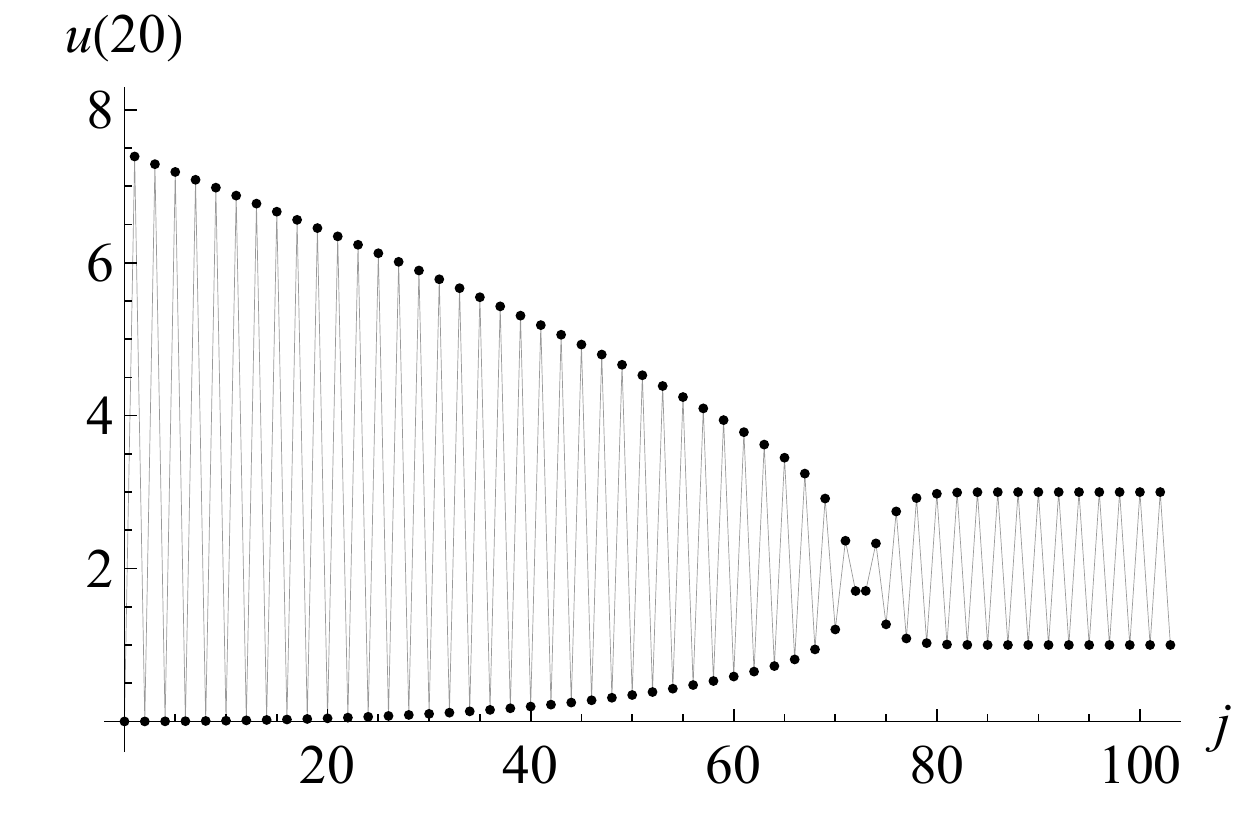}\qquad
\includegraphics[width=0.45\textwidth]{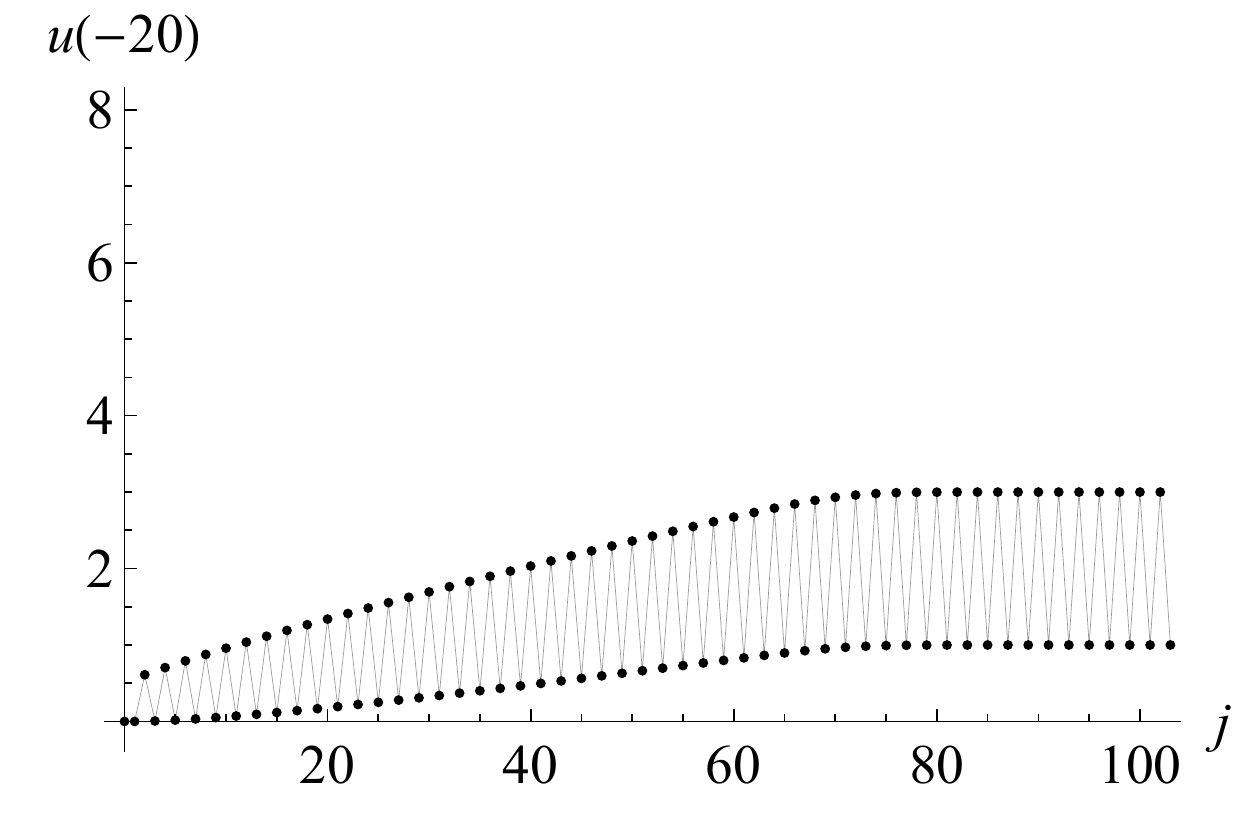}\\
\includegraphics[width=0.45\textwidth]{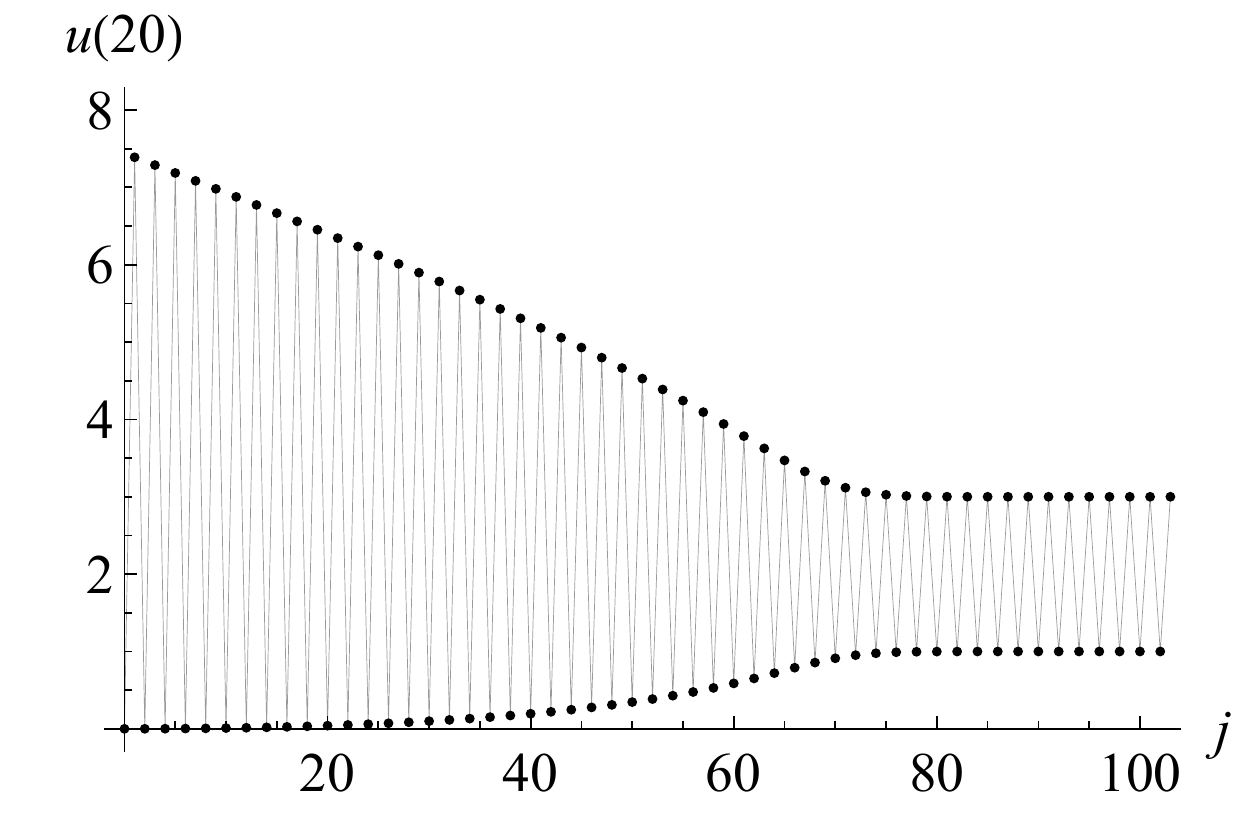}\qquad
\includegraphics[width=0.45\textwidth]{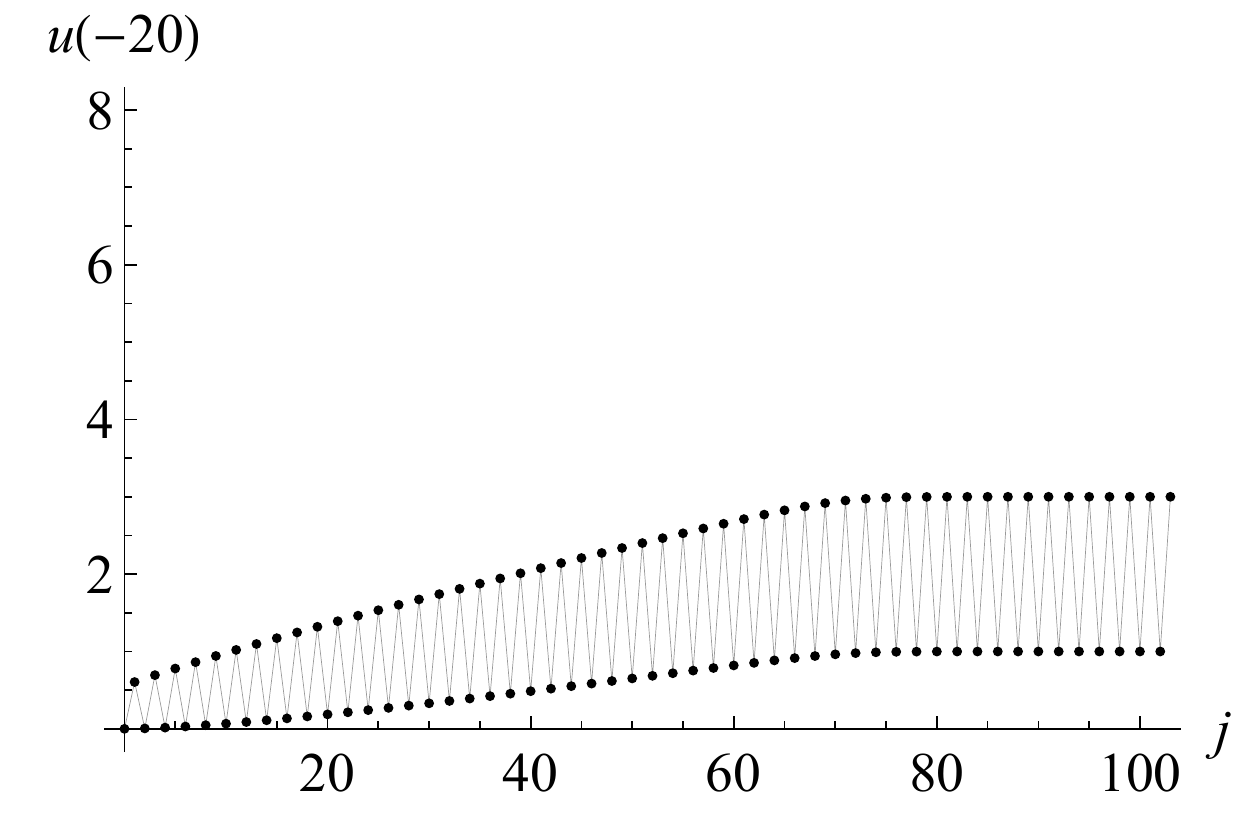}
\caption{Solution of the Volterra lattice with alternating initial data. Top: $0,1,3,1,3,\dots$, bottom: $0,3,1,3,1,\dots$.}
\label{fig:Nar1331} 
\end{figure}

The numeric solution corresponding to $q=c^2=3$ is shown on the top of Fig.~\ref{fig:Nar1331}. The limit values are
\begin{align*}
 & u_{2k-1}\to(c+1)^2,\quad u_{2k}\to0,\quad t\to+\infty,\\
 & u_{2k-1}\to(c-1)^2,\quad u_{2k}\to0,\quad t\to-\infty,\quad \text{if}~0<c\le1,\\
 & u_{2k-1}\to0,\quad u_{2k}\to(c-1)^2,\quad t\to-\infty,\quad \text{if}~1<c.
\end{align*}
The proof is based on the asymptotic formulae \cite{NIST}
\[
 {}_1F_{1}(a,b,z)=
 \left\{\begin{array}{cl}
  Ce^zz^{a-b}(1+C_1/z+o(z^{-1})), & z\to+\infty,\\[4pt]
  K(-z)^{-a}(1+K_1/z+o(z^{-1})), & z\to-\infty
 \end{array}\right. 
\]
(we do not need explicit expressions for $C,C_1,K$ and $K_1$) and also on the relation
\[
 \int^0_{-\infty}{}_1F_1\bigl(\tfrac{3}{2},3,4cx\bigr)e^{(c-1)^2x}\,dx=
  \left\{\begin{array}{cl}
   1,& 0<c\le1,\\
   c^{-2},& 1<c.
 \end{array}\right. 
\] 
This information is enough to calculate the power asymptotics for $u_1=f'_1/f_1$, then the formulae for remaining variables are recursively found directly from the lattice equations (in the case $t\to-\infty$, $1<c$, the variable $u_1$ decreases exponentially, and the power asymptotics starts from the variable $u_2=u'_1/u_1$).

Notice, that under the change $u_j(t)\to qu_j(qt)$ and $q\to1/q$, the inital data (\ref{stepq}) turn into $u_1=u_3=\dots=q$, $u_2=u_4=\dots=1$. The corresponding solution is shown on the bottom of Fig.~\ref{fig:Nar1331}. For this solution, the limit values for $t\to+\infty$ remain the same, but for $t\to-\infty$ the odd and even variables interchange, that is the vanishing variables are the odd ones if $c\le1$ and the even ones if $1<c$. 

%-------------------------------------------------------------------------------
\section{Reductions for $p>1$}\label{s:p}

%-------------------------------------------------------------------------------
\subsection{Constraints for the Bogoyavlensky lattice}

The BL$_p$ lattice
\begin{equation}\label{BLp}
 u'_j=u_j(u_{j+p}+\cdots+u_{j+1}-u_{j-1}-\cdots-u_{j-p})
\end{equation}
admits a wide set of non-autonomous constraints consistent with the dynamics. For example, the stationary equation for the sum of the higher and scaling symmetries $\partial_{t_2}+c(t\partial_t+\partial_{\tau_0})$ brings, after integration, to the constraint
\begin{equation}\label{constr0}
 u_j(u_{j-p}+\dots+u_{j+p})+cp(p+1)tu_j+cj+a_j+b_j=0,
\end{equation}
where $a_j$ and $b_j$ are integration constants satisfying the periodicity conditions $a_{j+p+1}=a_j$ and $b_{j+p}=b_j$. One can check that if $A_j$ is the expression in the left-hand side of (\ref{constr0}) then the differentiation by virtue of BL$_p$ satisfies the identity
\[
 A'_j=u_j(A_{j+p}+\dots+A_{j+1}-A_{j-1}-\dots-A_{j-p}),
\]
which proves the consistency with the lattice equations. As we already mentioned, equation (\ref{constr0}) for $p=1$ is equivalent to dP$_1$ equation and it reduces the Volterra lattice to the P$_4$ equation \cite{Fokas_Its_Kitaev_1991}. The reduction (\ref{constr0}) is interesting by itself, but it does not help to solve our problem of the decay of the unit step, since it does not admit these initial conditions. Another example is of the form
\begin{equation}\label{constr1}
 u_{j+1}\bigl(t(u_{j+1}+\cdots+u_{j+p+1})+b_{j+1}\bigr)-u_j\bigl(t(u_{j-p}+\cdots+u_j)+b_j-1\bigr)=0,
\end{equation}
with the quasi-periodic parameters $b_{j+p}=b_j+1$; here the left-hand side satisfies the equations
\begin{multline*}
\qquad A'_j=u_{j+1}(A_{j+1}+\cdots+A_{j+p})-u_j(A_{j-p}+\cdots+A_{j-1})\\
  +(u_{j+p}+\cdots+u_{j+1}-u_j-\cdots-u_{j+1-p})A_j.\qquad
\end{multline*}
If $p=1$ then equation (\ref{constr1}) turns into the constraint (\ref{cB}) with zero parameters $a$ and $c$. Again, there are no initial data in the form of the step in this case.

Other examples of polynomial constraints can be given, but finding one that would serve the required initial data is not so easy. We need an analogue of equation (\ref{cA}) which uses both scaling and master-symmetry, but the trouble is that for $p>1$ the master-symmetry becomes very complicated and practically unsuitable for calculations. It is nonlocal already for $p=2$ \cite{Zhang_Tu_Oevel_Fuchssteiner_1991}. Formally, the master-symmetry is defined for any $p$ as the result of applying the recursion operator $R$ to the scaling symmetry $u_{j,\tau_0}=u_j$:
\[
 u_{j,\tau_1}=R(u_j),
\]
but this operator itself is rather involved \cite{Wang_2012}:
\begin{equation}\label{R}
 R=u_j(T-T^{-p})(T-1)^{-1}\prod^{\overset{\curvearrowright}{p}}_{i=1}
  \bigl((T^{p+1-i}u_j-u_jT^{-i})(T^{p-i}u_j-u_jT^{-i})^{-1}\bigr).
\end{equation}
Application of the inverse operators contained in $R$ is equivalent to solving difference equations. If we apply $R$ to the right-hand side of (\ref{BLp}) then all these equations miraculously resolve in local form and lead to the higher symmetry. However, this is not the case if we apply $R$ to the scaling (except for the case $p=1$ for which we obtain the derivation (\ref{tau1})). In principle, the emerging nonlocal variables can be excluded from the constraint, but this leads to very cumbersome relations. For example, for $p=2$ this gives a 9-point constraint of degree 4 with respect to the fields $u_j$, and with the help of computer algebra it is possible to check that the hypergeometric solution from Proposition \ref{pr:step} actually satisfies this constraint. However, for $p=3$ such calculations seem hopeless.

%-------------------------------------------------------------------------------
\subsection{Constraint for the modified lattice}

Our approach is based on the transition to the modified Bogoyavlensky lattice
\begin{equation}\label{mBLpc}
 v'_j = (c-v_j)v_j(v_{j+1}\cdots v_{j+p} -v_{j-1}\cdots v_{j-p}),
\end{equation}
which is related with BL$_p$ by any of two Miura type substitutions (see e.g. \cite[Ch.\,17]{Suris_2003})
\begin{equation}\label{uvv}
 u_j=(c-v_{j-1})v_j\cdots v_{j+p-1}
\end{equation}
or
\begin{equation}\label{tuvv}
 u_j= v_j\cdots v_{j+p-1}(c-v_{j+p}).
\end{equation}
The homogeneous lattice with $c=0$ coincide with (\ref{mBL}) up to sign. Up to the author's knowledge, the recursion operator for (\ref{mBLpc}) is known only for $c=0$ \cite{Wang_2012}. It is similar to (\ref{R}) and the corresponding master-symmetry is also nonlocal for $p>1$. Moreover, if $c\ne0$ then there is no scaling symmetry to which we could apply the recursion operator. Therefore, at first glance, the new variables only make the problem more difficult. However, it turns out that the lattice (\ref{mBLpc}) admits a local reduction, which generalizes the example from Sect.~\ref{s:p1} and works for the step initial data.

\begin{proposition}
The lattice equation (\ref{mBLpc}) is consistent with $(2p+1)$-point difference equation
\begin{equation}\label{cG}
 G_j= t(c-v_j)\biggl(\sum^p_{s=0}v_{j-p+s}\cdots v_{j+s} 
   -c\sum^p_{s=1}v_{j-p+s}\cdots v_{j+s-1}\biggr) +b_jv_j -d_j=0,
\end{equation}
where parameters satisfy the quasi-periodic conditions
\begin{equation}\label{bd}
 b_{j+p}=b_j+1,\quad d_{j+p+1}=d_j+c.
\end{equation}
\end{proposition}

\begin{proof}
When checking the consistency, it is convenient (but not necessary) to take into account the substitution (\ref{uvv}) and rewrite the constraint in mixed variables as
\begin{equation}\label{vuu}
 G_j=t(v_j-c)(u_{j-p+1}+\dots+u_j)+tv_ju_{j+1}+b_jv_j-d_j=0,
\end{equation}
where $u_j$ and $v_j$ are differentiated in virtue of equations (\ref{BLp}) and (\ref{mBLpc}), respectively. The direct calculations prove the identity
\begin{equation}\label{F'}
 G'_j=(c-v_j)\bigl(T^{-1}+\dots+T^{-p}\bigr)\bigl(v_{j+1}\cdots v_{j+p}(G_{j+p+1}-G_j)\bigr),
\end{equation}
which means that if all $G_j=0$ then also $G'_j=0$, as required. 
\end{proof}

\begin{remark}
We can express $v_j$ from equation (\ref{vuu}) as a rational function of $u_{j-p+1},\dots,u_{j+1}$ and $t$. By substituting this into (\ref{uvv}) it is possible, in principle, to obtain an equivalent constraint of the form $H_j(t,u_{j-p},\dots,u_{j+p})=0$ in the original variables, but it is essentially more complicated.
\end{remark}

Due to the lattice (\ref{mBLpc}) and the constraint (\ref{cG}), each tuple $V_j=(v_j,\dots,v_{j+2p-1})$ satisfies a closed equation of $2r$ ODEs of first order, and equation (\ref{cG}) determines a rational map $V_j\mapsto V_{j+1}$ on solutions of these systems (B\"acklund transformation). The case $p=1$ is reduced to the Painlev\'e equations. It is natural to assume that the systems corresponding to $p>1$ also have the Painlev\'e property. However, as before, we do not analyse them in the general setting and restrict ourselves by studying a more special reduction that arises under the additional boundary condition $v_0=0$. In this case the lattice (\ref{mBLpc}) is divided into two subsystems for negative and positive $j$. They are on equal footing and we consider only $j>0$. Then the constraint equations with negative numbers are not needed since they enters into the right hand side of the identity (\ref{F'}) with the zero factor $v_0$. Thus, our reduction is defined by equations
\begin{equation}\label{v00}
 v_0=0,\quad G_0=G_1=G_2=\dots=0.
\end{equation}
The variables $v_1,\dots,v_p$ satisfy the equations
\begin{equation}\label{sysv1}
 v'_1=(c-v_1)v_1\cdots v_{p+1},~\dots,~~ v'_p=(c-v_p)v_p\cdots v_{2p},
\end{equation}
where $v_{p+1},\dots,v_{2p}$ can be eliminated by use of equations
\begin{equation}\label{sysv2}
 G_1(v_1,\dots,v_{p+1})=0,~\dots,~~ G_p(v_1,\dots,v_{2p})=0.
\end{equation}
Therefore, in contrast to the general reduction (\ref{cG}) which is equivalent to a system of ODEs of order $2p$, the termination on the half-line brings to a system of order $p$. We will show that this system is linearizable and its regular solution is expressed in terms of ${}_pF_p$-type hypergeometric functions.

Before proceeding to this problem, let us analyse the initial data provided by our reduction. A solution of the lattice equation (\ref{mBLpc}) with the constraint (\ref{cG}) is regular at $t=0$ if the values $v_j(0)$ satisfy the relations
\[
 b_jv_j(0)=d_j,
\]
where $b_j$ and $d_j$ satisfy equations (\ref{bd}). The boundary condition $v_0=0$ implies $d_0=0$, and $b_0$ can be arbitrary. If $c=0$ then we have $d_{j+p+1}=d_j$, but then some variables $v_{k(p+1)}$ also turn to zero, because all $b_{k(p+1)}$ can not vanish due to their quasi-periodicity. Therefore, in this case the lattice is restricted to a segment.  Further on, we assume that $c\ne0$. Additionally, we assume that
\[
 d_j\ne0,\quad b_j\ne 0,\quad j>0,
\]
then all $v_j(0)$ for $j>0$ are uniquely defined and are not equal to 0. These formulae determine the general family of inital data governed by the constraint (\ref{cG}). The initial data for the lattice (\ref{BLp}) are obtained under the substitutions (\ref{uvv}) or (\ref{tuvv}), and in both cases the condition $v_0=0$ provides $u_{-p+1}=\dots=u_0=0$, which is sufficient for the termination of (\ref{BLp}) onto the half-line. It is clear that for a given solution of (\ref{mBLpc}) these substitutions give two different solutions of BL$_p$, however, whole solutions set coincide. Indeed, when we use the first substitution, the initial data for $u_j$ are
\begin{equation}\label{icpbd}
 u_j(0)= \frac{(cb_{j-1}-d_{j-1})d_j\cdots d_{j+p-1}}{b_{j-1}\cdots b_{j+p-1}},\quad j>0,
\end{equation}
while the second substitution gives
\[
 u_j(0)= \frac{d_j\cdots d_{j+p-1}(cb_{j+p}-d_{j+p})}{b_j\cdots b_{j+p}},\quad j>0, 
\]
and this is reduced to (\ref{icpbd}) if we apply the change $b_i\to b_{i-1}$ and use the quasi-periodicity condition. Therefore one substitution is sufficient and in what follows we use only (\ref{uvv}) for which the calculations are slightly simpler (due to the broken symmetry when restricting onto the half-line). For the choice of $b_j=(j+1)/p$ and $d_j=cj/(p+1)$, the formula (\ref{icpbd}) gives our favourite step, up to a scaling factor: $u_j(0)=c^{p+1}p^p/(p+1)^{p+1}$.

\begin{remark}
Let us compare the reduction (\ref{cG}) for $p=1$ with (\ref{cA}). For $p=1$, the lattice equation (\ref{mBLpc}) admits a local master-symmetry and one can check that (\ref{cG}) is an integrated stationary equation for a non-autonomous symmetry:
\[
 tv_{j,t_2}+v_{j,\tau_1}+b_0v_{j,t} =(c-v_j)v_j(G_{j+1}-G_{j-1})=0,
\]
where flows in the left hand side are defined by equations
\begin{align*}
 v'_j &= (c-v_j)v_j(v_{j+1}-v_{j-1}),\\
 v_{j,t_2} &= (c-v_j)v_j\bigl((c-v_{j+1})v_{j+1}(v_{j+2}+v_j)-(c-v_{j-1})v_{j-1}(v_j+v_{j-2})\\
   &\qquad +c(c-v_{j-1}-v_{j+1})(v_{j-1}-v_{j+1})\bigr),\\
 v_{j,\tau_1} &= (c-v_j)v_j\bigl((j+1)v_{j+1}-(j-1)v_{j-1}-c\bigr).
\end{align*}
Moreover, both substitutions (\ref{uvv}) and (\ref{tuvv}) send this stationary equation to the constraint (\ref{cA}) for the Volterra lattice, up to renaming of parameters.
\end{remark}

%-------------------------------------------------------------------------------
\subsection{Linearization}

The system (\ref{sysv1}), (\ref{sysv2}) is linearized by introducing $\tau$-functions $f_j$ related with $u_j$ according to equations from Proposition \ref{pr:uf}. Although solutions of the lattice equation are determined by first $p$ functions $f_1,\dots,f_p$, it is convenient to use also $f_{p+1}=f'_1$. In proving the following theorem, we use the standard rules for differentiating generalized hypergeometric functions \cite{NIST}, which can be written as
\begin{equation}\label{dF}
\begin{aligned}
 & \Bigl(\frac{t}{a_j}\frac{d}{dt}+1\Bigr)F(A;B;t) = F(A+1_j;B;t),\\
 & \Bigl(\frac{t}{b_j}\frac{d}{dt}+1\Bigr)F(A;B+1_j;t) = F(A;B;t),\\
 & \frac{d}{dt}F(A;B;t)=\frac{a_1\cdots a_p}{b_1\cdots b_q}F(A+1_1+\dots+1_p;B+1_1+\dots+1_q;t),
\end{aligned}
\end{equation}
where $F(A;B;t)={}_pF_q(a_1,\dots,a_p;b_1,\dots,b_q;t)$ and $1_j$ denotes addition of 1 to $j$-th parameter.

\begin{theorem}\label{th:p}
Let parameters $a_j$, $b_j$ satisfy conditions
\begin{gather*}
 a_{j+p+1}=a_j+1,\quad b_{j+p}=b_j+1,\quad j\ge0,\\
 a_0=0,\quad c\ne0,\quad a_j\ne0,\quad b_j\ne 0,\quad j>0.
\end{gather*}
Then the solution of BL$_p$ with the boundary condition $u_j=0$ for $j\le0$ and initial data
\begin{equation}\label{icp}
 u_1(0)= c^{p+1}\frac{a_j\cdots a_{j+p-1}}{b_j\cdots b_{j+p-1}},\quad 
 u_j(0)= c^{p+1}\frac{(b_{j-1}-a_{j-1})a_j\cdots a_{j+p-1}}{b_{j-1}\cdots b_{j+p-1}},\quad j>1 
\end{equation}
is constructed by formulae (\ref{uf}) and (\ref{fL}) from functions
\begin{equation}\label{fH}
 f_j={}_pF_p(a_j,\dots,\widehat1,\dots,a_{j+p};b_j,\dots,b_{j+p-1};c^{p+1}t),\quad j=1,\dots,p,
\end{equation}
where $\widehat1$ denotes the omitted value $a_{p+1}=1$ in the first set of parameters.
\end{theorem}

\begin{proof} 
Initial data (\ref{icp}) are the same as (\ref{icpbd}), here we only denoted $d_j=ca_j$ and wrote the expression for $u_1$ separately, in order to handle the case $b_0=0$ (notice that if $j=1$ then the ratio $(cb_0-d_0)/b_0$ in (\ref{icpbd}) is obtained from $c-v_0$ and therefore it is always equal to $c$, even if $b_0=0$). Therefore, the solution $u_j$ is obtained by substitution (\ref{uvv}) from the solution $v_j$ of the lattice equation (\ref{mBLpc}) which satisfies the constraint (\ref{cG}) and the boundary condition $v_0=0$. 

We derive relations between $v_j$ and $f_j$. First, we use the constraint equation in the form (\ref{vuu}). By replacing $u_j$ through $f_j$ and taking into account that $u_j=0$ for $j\le0$ we obtain
\begin{equation}\label{vf}
 v_j\Bigl(t\frac{f'_{j+1}}{f_{j+1}}+b_j\Bigr)=c\Bigl(t\frac{f'_j}{f_j}+a_j\Bigr),\quad j=1,\dots,p,
\end{equation}
which gives $v_1,\dots,v_p$. On the other hand, we have from the substitution (\ref{uvv}) 
\begin{gather*}
 u_1=\frac{f_{p+1}}{f_1}=cv_1\cdots v_p,\\
 u_j=\frac{f'_j}{f_j}-\frac{f'_{j-1}}{f_{j-1}}=(c-v_{j-1})v_j\cdots v_{j+p-1}
    =\frac{v'_{j-1}}{v_{j-1}},\quad j=2,\dots,p+1,
\end{gather*}
whence it follows
\[
 v_j=\kappa_j\frac{f_{j+1}}{f_j},\quad j=1,\dots,p,\quad c\kappa_1\cdots \kappa_p=1.
\]
Integration constants $\kappa_j$ are fixed by initial data $v_j(0)=ca_j/b_j$ and $f_1(0)=\dots=f_p(0)=1$, which gives
\begin{equation}\label{vff}
 v_j=\frac{ca_jf_{j+1}}{b_jf_j},\quad j=1,\dots,p-1,\quad
 v_p=\frac{b_1\cdots b_{p-1}f_{p+1}}{c^pa_1\cdots a_{p-1}f_p}.
\end{equation}
Substituting here (\ref{vf}), we obtain a closed system of linear differential equations for $f_j$:
\begin{equation}\label{fsystem}
\begin{gathered}
  f'_1=f_{p+1},\quad \frac{1}{a_j}tf'_j+f_j = \frac{1}{b_j}tf'_{j+1}+f_{j+1},\quad j=1,\dots,p-1,\\
 c^{p+1}\frac{a_1\cdots a_p}{b_1\cdots b_p}\Bigl(\frac{1}{a_p}tf'_p+f_p\Bigr) = \frac{1}{b_p}tf'_{p+1}+f_{p+1},
\end{gathered}
\end{equation} 
whose solution leads to the answer (\ref{fH}). Indeed, let us denote
\[
 F(A;B;t)={}_pF_p(a_1,\dots,a_p;b_1,\dots,b_p;t)
\]  
then taking into account that $a_{j+p+1}=a_j+1$ and $b_{j+p}=b_j+1$ (and also that the parameters in each set $A$ and $B$ can be permuted) functions (\ref{fH}) can be written as
\begin{gather*}
 f_1=F(A;B;c^{p+1}t),\quad f_2=F(A+1_1;B+1_1;c^{p+1}t),~~\dots\\
 f_j=F(A+1_1+\cdots+1_{j-1};B+1_1+\cdots+1_{j-1};c^{p+1}t),~~ \dots,\\
 f_p=F(A+1_1+\cdots+1_{p-1};B+1_1+\cdots+1_{p-1};c^{p+1}t).
\end{gather*}
From the identities (\ref{dF}) it follows that $f_{p+1}=f'_1=c^{p+1}\frac{a_1\cdots a_p}{b_1\cdots b_p}F(A+1_1+\dots+1_p;B+1_1+\dots+1_p;c^{p+1}t)$ and all these functions satisfy the system (\ref{fsystem}).
\end{proof}

Proposition \ref{pr:step} is a special case of the theorem proved corresponding to the choice of parameters
\begin{equation}\label{abc}
 a_j=j/(p+1),\quad b_j=(j+1)/p,\quad c^{p+1}=(p+1)^{p+1}/p^p.
\end{equation}
Notice that in this case the factor on the left hand side of the last equation (\ref{fsystem}) is equal to 1.

\begin{remark}
If parameters $a_j$ and $b_j$ coincide for some $j$ then the hypergeometric functions degenerate, and the variable $u_{j+1}$ is zero, according to (\ref{icp}). For $p>1$ this does not lead to the restriction of the lattice to a segment.
\end{remark}

%-------------------------------------------------------------------------------
\subsection{Asymptotics for the unit step}\label{s:asymp}

Until now, we have not been interested in the global properties of the solutions found. Series of type ${}_pF_p(A;B;z)$ converge in the entire complex plane of $z=c^{p+1}t$, so the $\tau$-functions $f_j$ defined from them are entire (for all $j$, since for $j>p+1$ they are written in terms of Wronskians). However, functions $u_j$ and $v_j$ can have poles. Apparently, for real-valued solutions, the most important condition of regularity is the constancy of the sign of the initial data \cite{Kulaev_Shabat_2019}. An analysis of the restrictions on parameters that arise here is beyond the scope of this article. We confine ourselves to calculating the first term of the asymptotics for the case of the unit step which is as a typical example defined for all $t\in\mathbb R$ (for $p=1$, a slightly more detailed description can be found in \cite{Adler_Shabat_2019}, where we obtained two terms of the asymptotics of solutions from Proposition \ref{pr:p1}; moreover, some numeric experiments described there show that small changes in parameters and initial data do not change the general structure of solutions).

We use the following asymptotic formulae for the generalized hypergeometric functions \cite{NIST}:
\begin{align}
\label{Fas+}
 {}_pF_p(a_1,\dots,a_p;b_1,\dots,b_p;z)&=Ce^zz^\nu(1+C_1/z+\dots),\quad z\to+\infty,\\
\label{Fas-}
 {}_pF_p(a_1,\dots,a_p;b_1,\dots,b_p;z)&=\sum^p_{j=1}K_j(-z)^{-a_j}(1+K_{j,1}/z+\dots),\quad z\to-\infty,
\end{align}
where $\nu=a_1+\cdots+a_p-b_1-\cdots-b_p$. The expressions for remaining constants and further coefficients of the expansions in terms of $a_j$ and $b_j$ are also known, but we do not need their explicit form.

\begin{proposition}\label{pr:asymp}
The solution of the lattice equation (\ref{BL}) with the initial data (\ref{step}) has the following asymptotics:
\begin{gather}
\label{+inf}
 u_{(p+1)k+i} = \left\{\begin{array}{ll}
  \dfrac{(p+1)^{p+1}}{p^p}-\dfrac{4k+3}{2t}+o(t^{-1}), & i=1,\\[8pt]
  \dfrac{p^{p-1}(k+1)\bigl(k+\frac{p+i}{p+1}\bigr)}{(p+1)^{p+1}t^2}+o(t^{-2}), & i=2,\dots,p+1
\end{array}\right. ~\text{for}~t\to+\infty,\\
\label{-inf}
 u_{pk+i} = -\frac{k+1}{(p+1)t}+o(t^{-1}),\quad i=1,\dots,p\quad \text{for}~t\to-\infty.
\end{gather}
\end{proposition}

\begin{proof} 
Consider first the asymptotics for $t\to+\infty$. The solution is constructed by functions (\ref{fH}) with the parameters (\ref{abc}). It is easy to check that these parameters satisfy, for all $j$ and $p$, the equality
\[
 \nu=a_j+\cdots+a_{j+p}-1-b_j-\cdots-b_{j+p-1}= -3/2,
\] 
which defines one and the same parameter $\nu$ in (\ref{Fas+}) for all functions $f_1,\dots,f_{p+1}$. Then (\ref{uf}) implies that first $p+1$ functions have the asymptotics of the form
\[
 u_1=c^{p+1}-\frac{3}{2t}+o(t^{-1}),\quad u_j=\frac{\beta_j}{t^2}+o(t^{-2}),\quad j=2,\dots,p+1,
\]
where $\beta_j$ are related with the coefficients $C_1$ from the respective expansions (\ref{Fas+}). A direct substitution into (\ref{BL}) proves that this pattern is reproduced for subsequent $j$ modulo $p+1$:
\[
 u_j=\left\{\begin{array}{cl}
  c^{p+1}-\alpha_j/t+\dots, & j\equiv 1\mod(p+1),\\
  \beta_j/t^2+\dots,        & j\not\equiv 1\mod(p+1).
 \end{array}\right.
\]
Moreover, $\alpha_j$ satisfy the recurrence relation $\alpha_{j+p+1}=\alpha_j-2$, which gives the first line in (\ref{+inf}), but for $\beta_j$ there is only one equation $c(\beta_{j+1}+\cdots+\beta_{j+p}-\beta_{j-p}-\cdots-\beta_{j-1})=\alpha_j$, $j\equiv 1\mod(p+1)$, which is not sufficient when $p>1$. However, we should remember that our solution satisfies also the constraint (\ref{cG}) which gives additional information. Relations (\ref{vff}) and the lattice equation (\ref{mBLpc}) imply the following asymptotic formula for the variables $v_j$:
\[
 v_j=\left\{\begin{array}{cl}
  \gamma_j/t+\dots,   & j\equiv 0\mod(p+1),\\
  c-\delta_j/t+\dots, & j\not\equiv 0\mod(p+1).
 \end{array}\right.
\]
By substituting formulae for $u_j$ and $v_j$ into the constraint equation (\ref{vuu}) and keeping only the leading terms we obtain
\[
 \gamma_j=a_j/c^p,\quad \delta_j=(b_j-a_j)/c^p,
\]
where $a_j$ and $b_j$ are defined by (\ref{abc}). In turn, $\alpha_j$ and $\beta_j$ are expressed through $\gamma_j$ and $\delta_j$ by expanding the relation $u_j=(c-v_{j-1})v_j\cdots v_{j+p-1}$, which gives
\begin{gather*}
 \alpha_j = c^p(\gamma_{j-1}+\delta_j+\cdots+\delta_{j+p-1}),\\
 \beta_j=c^{p-1}\gamma_{(p+1)(k+1)}\delta_{j-1},\quad j=(p+1)k+i,\quad i=2,\dots,p+1,
\end{gather*}
and we arrive to (\ref{+inf}) after easy manipulations.

For $t\to-\infty$, we use the formula (\ref{Fas-}). The leading term of the asymptotics of $f_j$ corresponds to the minimal value in the set of parameters, that is $a_j=j/(p+1)$, and equation (\ref{uf}) provides the asymptotics for first $p$ variables
\[
 u_j=-\frac{1}{(p+1)t}+o(t^{-1}),\quad u_j=1,\dots,p.
\]
In this case it is sufficent to use the lattice equation only. It proves that $u_j=\alpha_j/t+o(t^{-1})$ for all $j$ and that $\alpha_{j+p}+\cdots+\alpha_{j+1}-\alpha_{j-1}-\cdots-\alpha_{j-p}=-1$, which leads to the formula (\ref{-inf}).
\end{proof}

\begin{figure}[t]
\centering
\includegraphics[width=0.46\textwidth]{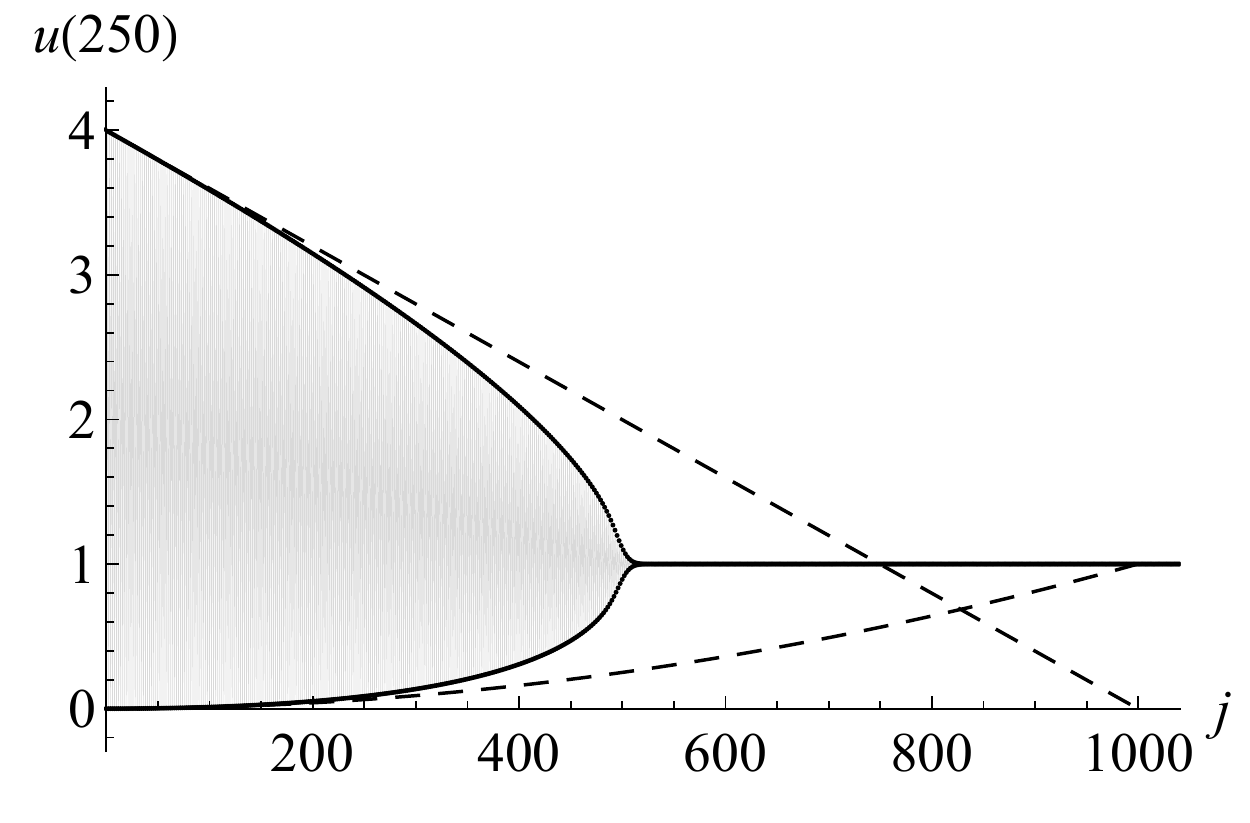}\qquad
\includegraphics[width=0.46\textwidth]{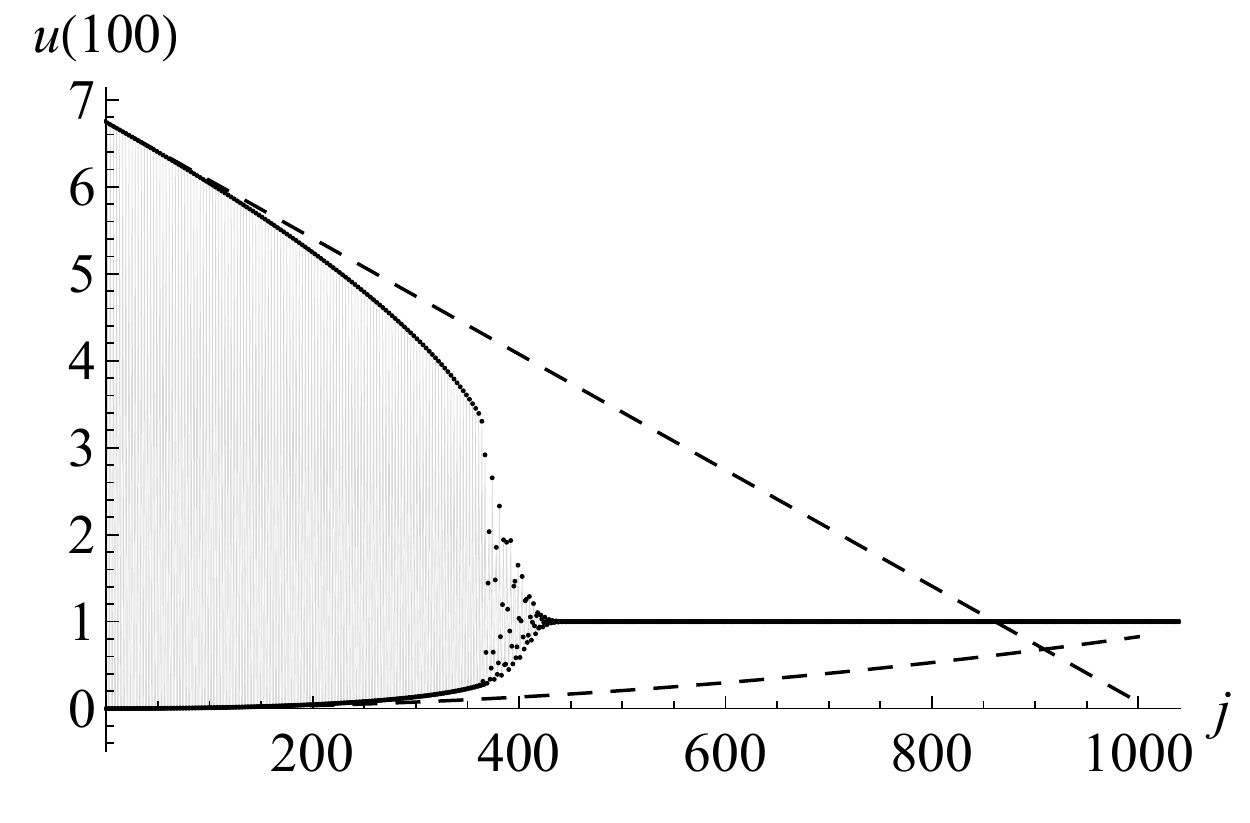}
\caption{Decay of the unit step for $p=1$ and $p=2$. The dashed lines show the asymptotics (\ref{+inf}) as functions of $j$ for a fixed $t$.}
\label{fig:p12} 
\end{figure}

The formula (\ref{-inf}) gives a good approximation of the solution for $t<0$ up to the values $j\sim p(p+1)t$. Near this point, the solution smoothly joins the initial unit step. Slightly coarsening, we can say that for $t<0$ the solution behaves almost like a piecewise linear function $u_j(t)=\min\{1,j/(p(p+1)t)\}$. The decay of the step for $t>0$ is more interesting. Here the formula (\ref{+inf}) gives the upper limit of the variables, and also makes it possible to roughly estimate the size of the decay zone. To do this, consider the leading terms of the asymptotics (\ref{+inf}) as functions of $j$ for a fixed $t>0$. The plots of these functions are shown on Fig.~\ref{fig:p12} (for $p>1$, parabolas corresponding to the second line in (\ref{+inf}) practically coincide). We see that the main information is given by the straight line
\[
 u=\frac{(p+1)^{p+1}}{p^p}-\frac{1}{t}\Bigl(2\frac{j-1}{p+1}+\frac{3}{2}\Bigr)
\] 
below which lies the graph of the solution. This line meets the $j$-axis near the point $j=Ct$, where $C=(p+1)^{p+2}/(2p^p)$
is an upper estimate for the rate of the decay zone expansion. It can be seen from the plots that the actual speed $j(t)$ of the decay point (which can be defined as the maximum value of $j$ at which the value of $u_j(t)$ differs noticeably from 1) is several times less. Somewhat more accurate estimates can be obtained by calculating the next asymptotic terms, but this does not completely solve the problem, since the asymptotics describes the solution behaviour for a fixed $j$ far beyond the decay point. For $p=1$, the even and odd variables separate immediately after passing through this point and tend monotonically to their limit values. However, already for $p=2$, one can notice a small transition region formed behind the decay point, in which the solution behaves quite chaotically (that is, the variables separate not immediately, but after several oscillations). As $p$ increases, the relative size of this transition zone increases (see Fig.~\ref{fig:p37}), and its description remains an open problem.

\begin{figure}[t]
\centering
\includegraphics[width=0.46\textwidth]{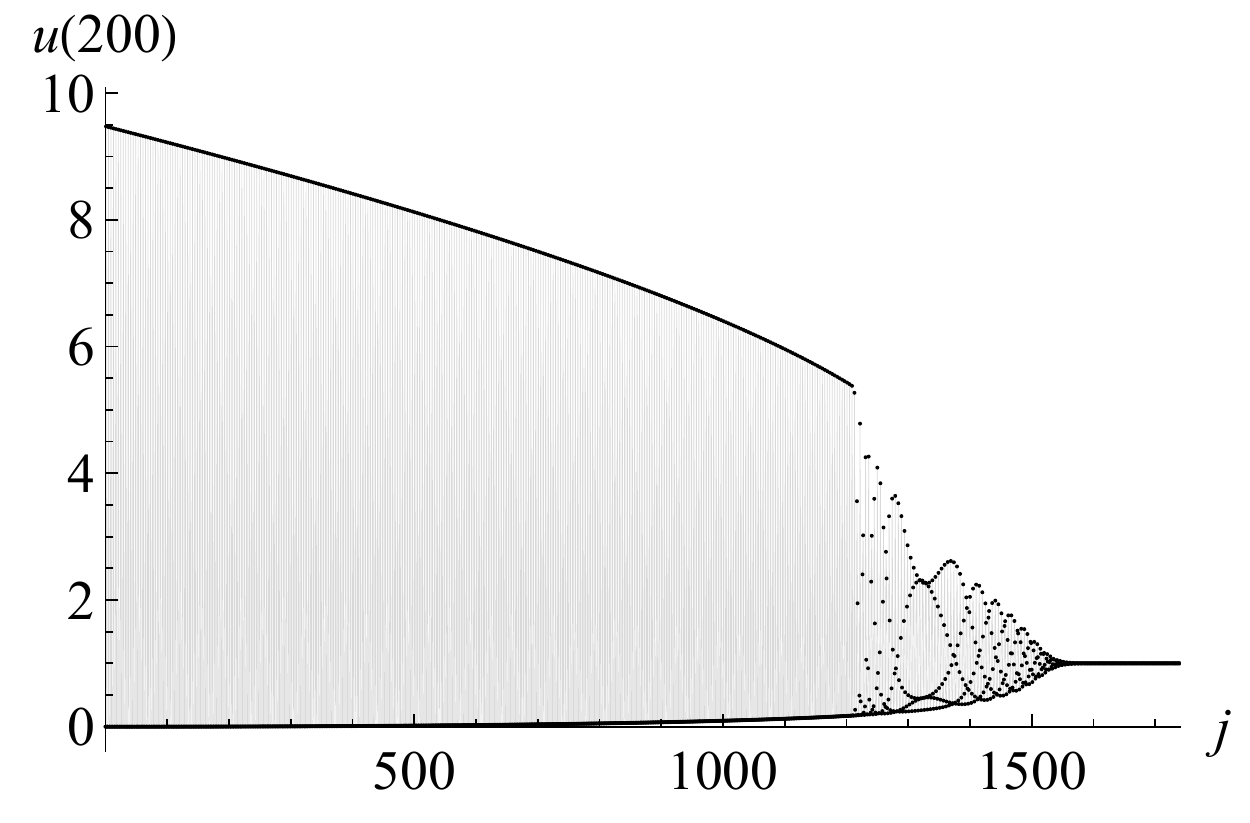}\qquad
\includegraphics[width=0.46\textwidth]{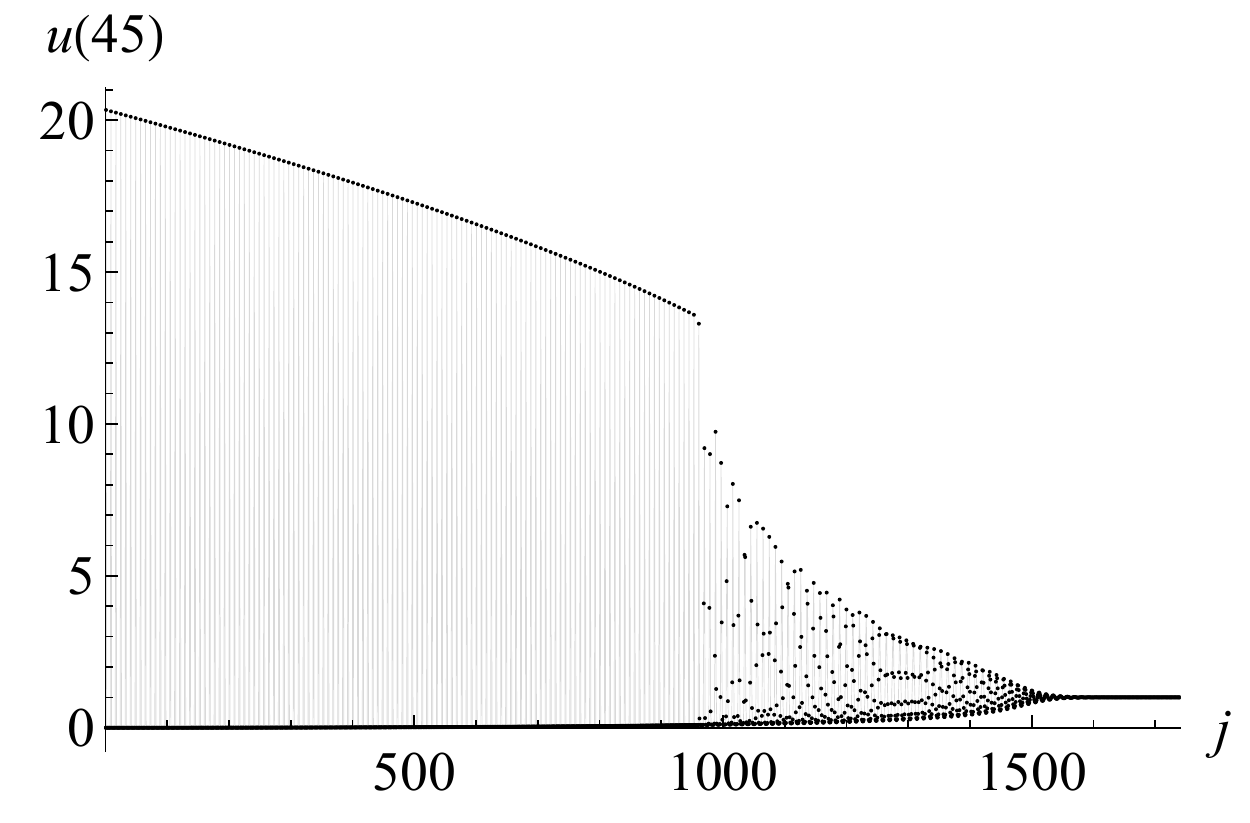}
\caption{Decay of the unit step for $p=3$ and $p=7$.}
\label{fig:p37} 
\end{figure}

%-------------------------------------------------------------------------------
\section{Conclusion}

The main idea of our work is that simple experiments with Taylor series can lead to the discovery of unexpected exact solutions of integrable equations. Indeed, if it were not for the observation with the generalized Catalan numbers (made possible thanks to the extensive \cite{OEIS} database) then, most likely, our solutions would not have been found soon, since a priori there is no reason to expect that certain initial data correspond to solutions described by some ODEs. On the other hand, the method of symmetry reductions, common in the theory of integrable equations, is more systematic and, as we have seen, can be used in the opposite direction, to prove identities that arise in the combinatorics from completely different considerations. It is to be hoped that the list of such examples can be extended.

In the problem of decay of the unit step, many questions remain open. Firstly, this is a more accurate description of the decay zone, and secondly, generalizations for more complex initial data. In particular, the general stationary solution of  BL$_p$ is of the form $u_j=\alpha_{j\bmod p}+\beta_{j\bmod(p+1)}$ and the question arises about its decay for the terminated lattice. As we have seen in the example of $p=1$ (Narayana polynomials), this problem is solved by passing from the Volterra lattice to the Toda lattice. Similar substitutions are known which relates BL$_p$ with the multifield Toda lattice \cite{Suris_2003} and the Belov--Chaltikian lattice \cite{Belov_Chaltikian_1993, Hikami_Sogo_Inoue_1997}. Moreover, the Belov--Chaltikyan lattice is known to have local master-symmetry \cite{Khanizadeh_Mikhailov_Wang_2013}, so it probably can be used to obtain non-autonomous reductions covering a fairly rich family of initial data.

\subsection*{Declarations}

\subsubsection*{Funding}

The work was done at Ufa Institute of Mathematics with the support by the grant \#21-11-00006 of the Russian Science Foundation, https://rscf.ru/project/21-11-00006/.

\subsubsection*{Data Availability Statements}

All data generated or analysed during this study are included in this published article.

%-------------------------------------------------------------------------------

\end{document}